\documentclass[11pt]{article}
\usepackage[margin=1in]{geometry}
\usepackage{graphicx}
\usepackage{amsfonts}
\usepackage{hyperref}
\usepackage{amsthm}
\usepackage{amssymb}
\usepackage{mathtools}
\usepackage[utf8]{inputenc}
\usepackage[english]{babel}
\usepackage[dvipsnames]{xcolor}
\usepackage[capitalize]{cleveref}
\usepackage[b]{esvect}
\usepackage{thmtools}
\usepackage{thm-restate}

\usepackage{tikz} 
\usepackage{mathdots}
\usetikzlibrary{decorations.pathreplacing,calc,arrows,arrows.meta,patterns,positioning,backgrounds,shapes.geometric}

\newcommand{\bigO}{\mathcal{O}}

\definecolor{corlinks}{RGB}{0,0,150}
\definecolor{cormenu}{RGB}{0,0,150}
\definecolor{corurl}{RGB}{200,0,0}

\hypersetup{
colorlinks=true,
urlcolor=corurl,
linkcolor=corlinks,
menucolor=cormenu,
citecolor=corlinks,
pdfborder= 0 0 0
}

\def\colorful{1}
\ifnum\colorful=1

\fi
\ifnum\colorful=0

\fi

\newcommand{\eps}{\varepsilon}
\newcommand{\Bits}{\{0,1\}}
\newcommand{\poly}{\mathsf{poly}}
\newcommand{\shortpath}{\mathsf{SHORT}}
\newcommand{\longpath}{\mathsf{LONG}}

\newcommand{\polylog}[1]{\mathsf{poly}(\log(#1))}
\newcommand{\stconn}{\mathsf{STCONN}}
\newcommand{\CSPACE}{\mathsf{CSPACE}}
\newcommand{\CTS}{\mathsf{CTIMESPACE}}

\newcommand{\ed}{\mathsf{ED}}
\newcommand{\discreteFrechet}{\ensuremath{\mathsf{DFD}}}
\newcommand{\lcs}{\mathsf{LCS}}
\newcommand{\layergrid}{\mathsf{Grid}}
\newcommand{\frechet}{\mathsf{Frech}}
\renewcommand{\mod}{\,\text{mod}\,}
\newcommand{\inneralg}{\mathsf{{INNER}}}
\newcommand{\outeralg}{\mathsf{{OUTER}}}
\newcommand{\codecomment}[1]{\hfill{}\hfill{}{\color{darkgray} //~#1}\hfill{}}
\renewcommand{\codecomment}[1]{\bgroup\hfill\makebox[3.1in][l]{\small//~#1}\egroup}

\newcommand{\weightoracle}{\mathcal{D}}

\newcommand{\calH}{\mathcal{H}}

\newcommand{\calW}{\mathcal{W}}
\newcommand{\calU}{\mathcal{U}}

\newcommand{\N}{\mathbb{N}}
\newcommand{\Z}{\mathbb{Z}}

\newcommand{\F}{\mathbb{F}}

\newcommand{\vin}{v_{in}}
\newcommand{\vout}{v_{out}}
\newcommand{\vmid}{v_{mid}}
\newcommand{\bigvin}{V_{in}}
\newcommand{\bigvout}{V_{out}}
\newcommand{\bigvmid}{V_{mid}}
\newcommand{\cin}{c_{in}}
\newcommand{\cout}{c_{out}}
\newcommand{\cmid}{c_{mid}}
\newcommand{\bigRin}{\vv{r_{in}}}
\newcommand{\bigRout}{\vv{r_{out}}}
\newcommand{\bigRmid}{\vv{r_{mid}}}
\newcommand{\bigRinprime}{\vv{r'_{in}}}
\newcommand{\bigRoutprime}{\vv{r'_{out}}}
\newcommand{\bigRone}{\vv{r_1}}
\newcommand{\bigRtwo}{\vv{r_2}}
\newcommand{\bigRthree}{\vv{r_3}}
\newcommand{\tauone}{\vv{\tau_1}}
\newcommand{\tautwo}{\vv{\tau_2}}

\newcommand{\tauin}{\vv{\tau_{in}}}
\newcommand{\tauout}{\vv{\tau_{out}}}
\newcommand{\taumid}{\vv{\tau_{mid}}}
\newcommand{\lin}{\ell_{in}}
\newcommand{\lout}{\ell_{out}}
\newcommand{\lmid}{\ell_{mid}}
\newcommand{\hclass}{h_{class}}
\newcommand{\itop}{i_{top}}
\newcommand{\iin}{i_{in}}
\newcommand{\iout}{i_{out}}

\newcommand{\pluseq}{\mathrel{+}=}
\newcommand{\minuseq}{\mathrel{-}=}

\newtheorem{theorem}{Theorem}[section]
\newtheorem{lemma}[theorem]{Lemma}
\newtheorem{proposition}[theorem]{Proposition}

\newtheorem{definition}[theorem]{Definition}
\newtheorem{remark}[theorem]{Remark}

\crefname{theorem}{Theorem}{Theorems}
\crefname{lemma}{Lemma}{Lemmas}
\crefname{definition}{Definition}{Definitions}
\crefname{proposition}{Proposition}{Propositions}
\crefname{corollary}{Corollary}{Corollaries}
\crefname{remark}{Remark}{Remarks}

\newcommand*\samethanks[1]{\footnotemark[#1]}

\title{Frontier Space-Time Algorithms Using Only Full Memory}

\author{Petr Chmel\thanks{Partially supported by the Grant Agency of the Czech Republic under the grant agreement no. 24-10306S, by GAUK project No. 70924 and by Charles Univ. project UNCE 24/SCI/008.} \\ Charles University \\ \texttt{chmel@iuuk.mff.cuni.cz} \and Aditi Dudeja\thanks{This work was initiated while the author was affiliated with the University of Salzburg. This project has received funding from the European Research Council (ERC) under the European Union’s Horizon 2020 research and innovation programme (grant agreement No 947702)}\\ The School of Data Science \\ The Chinese University of Hong Kong, Shenzhen \\ \texttt{aditidudeja@cuhk.edu.cn} \and Michal Kouck\'y\thanks{Partially supported by the Grant Agency of the Czech Republic under the grant agreement no. 24-10306S and by Charles Univ. project UNCE 24/SCI/008.} \\ Charles University \\ \texttt{koucky@iuuk.mff.cuni.cz} \and Ian Mertz\samethanks{3} \\ Charles University \\ \texttt{iwmertz@iuuk.mff.cuni.cz} \and Ninad Rajgopal\samethanks{3} \\ Charles University \\ \texttt{ninad@iuuk.mff.cuni.cz}}

\date{\today}

\begin{document}

\maketitle
\begin{abstract}
    We develop catalytic algorithms for fundamental problems in algorithm design that run in polynomial time, use only $\bigO(\log(n))$ workspace, and use sublinear catalytic space matching the best-known space bounds of non-catalytic algorithms running in polynomial time.
    
    First, we design a polynomial time algorithm for directed $s$-$t$ connectivity using $n \big/ 2^{\Theta(\sqrt{\log n})}$ catalytic space, which matches the state-of-the-art time-space bounds in the non-catalytic setting~\cite{BBRS98}, and improves the catalytic space usage of the best known algorithm~\cite{cook_pyne_26_efficient}. Furthermore, using only $\bigO(\log(n))$ random bits we get a randomized algorithm whose running time nearly matches the fastest time bounds known for space-unrestricted algorithms.

    Second, we design polynomial time algorithms for the problems of computing Edit Distance, Longest Common Subsequence, and the Discrete Fr\'{e}chet Distance, again using $n \big/ 2^{\Theta(\sqrt{\log n})}$ catalytic space. This again matches non-catalytic time-space frontier for Edit Distance and Least Common Subsequence~\cite{KiyomiHO21}.
\end{abstract}
\tableofcontents

\section{Introduction}
\label{sec:intro}

\subsection{Catalytic Computation}
\label{sec:intro:catalytic}
Recently, a new paradigm for space-bounded computation called \emph{catalytic computing} has received a great deal of attention.
In short, a catalytic machine is one whose memory is almost completely full with arbitrary data at the start; while it may freely use this full memory, known as \emph{catalytic space}, in a read-write manner as with ordinary space, it must restore the initial data to the catalytic memory at the end of the computation.
Catalytic computing was defined by Buhrman, Cleve, Kouck{\'y}, Loff, and Speelman, in  \cite{BCKLS14}, who showed that the addition of polynomial catalytic space to a log-space machine gives it significant power, more so even than non-determinism, randomness, or both.

Since then there have been a host of results regarding the catalytic model, such as structural results~\cite{BuhrmanKLS18,GuptaJST19,DattaGJST20,CookM22,CookLMP25,KouckyMPS25}, placing further problems into the model~\cite{Pyne25,AgarwalaM25,AlekseevFMSV25,AgarwalaAV26}, and alternative catalytic models~\cite{BisoyiDS22,GJST24,FolkertsmaMST25,BuhrmanFMSSST25,BisoyiDRS25,CGMPS25} (see surveys of Kouck\'y~\cite{Koucky16} and Mertz~\cite{Mertz23} for more background).
Furthermore, the techniques of using full memory were also used in the context of ordinary space-bounded computation, giving rise to a novel paradigm for derandomization~\cite{DoronPT24,LiPT24,DoronPTW25,GTS25} as well as shedding new light on the relationship of space and time~\cite{CookM20,CM24_tree_eval,Wil25_simulate,Shalunov25}.

While all problems in $\mathsf{NL}$ are computable in polynomial time $\mathsf{CL}$ (i.e., the class of problems computable by a log-space machine with polynomial catalytic space), the approach of \cite{BCKLS14} incurs a large overhead in both time and catalytic space. In response to this, there been a recent focus, responding to a question posed by Mertz~\cite{Mertz23}, of truly efficient catalytic algorithms whose time and catalytic space usage are both comparable to the best known results in the non-catalytic setting. Cook and Pyne~\cite{cook_pyne_26_efficient} showed that connectivity in graphs ($\stconn$) can be solved in quadratic time using linear catalytic space; they also show a similar result for estimating random walks. These algorithms are conceptually simple enough to be implemented in practice---to wit, they are only a $\polylog{n}$ factor slower than breadth-first search---and redevelop pre-existing techniques from the catalytic literature in new and streamlined ways.

In this work we push the algorithmic frontier of connectivity and other problems, by asking whether we can capture the best \emph{space-time} algorithms, namely those using (reasonable) polynomial time and sublinear space, entirely within the catalytic paradigm.

\subsection{Sublinear-Space Polynomial-Time Algorithms}
\label{sec:intro:problems}
Connectivity is an extremely fundamental primitive in graph algorithms, both in theory and in practice. In complexity theory the many variants of $\stconn$ characterize logspace and their adjacent classes, while on the algorithmic side there has been an extensive history of studying its space, time, space-time trade-offs and heuristics.

While the straightforward DFS and BFS algorithms for $\stconn$ use linear time in the number of edges, the space requirement for these algorithms is linear in the number of vertices.
In contrast, Savitch's theorem shows that this problem can be solved in $\bigO(\log^2(n))$ space and time $n^{\bigO(\log(n))}$ \cite{Sav70}.
To meet both conditions, i.e. run in polynomial time and use sublinear space, a result by Barnes, Buss, Ruzzo, and Schieber~\cite{BBRS98} gives space complexity $\frac{n}{2^{\Theta{\left( \sqrt{\log(n)} \right)}}}$.

Except for restricted graph classes~\cite{AllenderL98,KannanKR08,ReingoldTV06,Reingold08}, no better results are known. Furthermore, the algorithm of~\cite{BBRS98} can be implemented in a restricted model of computation called Node Naming Jumping Automata on Graphs (NNJAG), for which we have a matching space lower bound in the polynomial time regime due to Edmonds, Poon, and Achlioptas~\cite{EPA99_nnjag_lb}, even for probabilistic algorithms.

There are also a number of fundamental algorithmic questions which reduce to connectivity, including Edit Distance ($\ed$), Longest Common Subsequence ($\lcs$), and Discrete Fr\'{e}chet Distance ($\discreteFrechet$). The first two, which quantify how close two strings are in various ways, are very well-studied problems with applications in various areas such as data analysis, computational biology, and text processing~\cite{BoroujeniEGHS18,AbboudBW15,AndoniKK10,AndoniK09,BackursI15,AbboudHWW16,AbboudR18,AlvesCS06,AndoniDGIR03,AndoniGMP13,AndoniK07,AndoniK08,BansalLMZ10,BatuFC06,BoroujeniS19,BringmannK18,BringmanGSW16,BringmanK15}.
Meanwhile, 
$\discreteFrechet$ has been studied extensively in computational geometry, motivated by applications such as dynamic time-warping~\cite{KeoghP99}, matching of time series in databases~\cite{KimKS05}, and analysis of moving objects~\cite{BuchinBG08,BuchinBGLL11}. 

For $\ed$ and $\lcs$, the classic dynamic programming algorithms need linear space, and similarly for $\discreteFrechet$, the original dynamic programming algorithm of Eiter and Mannila~\cite{EiterM94} takes $\bigO(mn)$ time (where $m$ and $n$ are numbers of vertices of the polygonal chains); this was subsequently improved to $\bigO(n)$ space and $\bigO \left( \frac{mn\log\log(n)}{\log(n)} \right)$ time by Agarwal et al.~\cite{AgarwalARKS14}.
On the other hand, if we allow for $n^{\bigO(\log(n))}$ runtime, then we can solve these problems in $\bigO(\log^2 n)$ space via a reduction to $\stconn$.
With regards to space-time results, their respective reductions to $\stconn$ involve a quadratic blowup in the input size $n$, which would render the result of~\cite{BBRS98} useless; however, a recent work of Kiyomi, Horiyama, and Otachi~\cite{KiyomiHO21} gives a polynomial time algorithm for $\lcs$ and $\ed$ using $\frac{n}{2^{\Theta\left( \sqrt{\log(n)} \right)}}$ space.
While better algorithms are known in the approximation and/or streaming regimes~\cite{ChakrabortyGK16,AndoniKO10,SaksS13,ChengFHJLRSZ21}, any improvements for these problems beyond the $\stconn$ threshold seem firmly out of reach~\cite{ChengFHJLRSZ21}.

\subsection{Our Results}
\label{sec:intro:results}
In this work, we develop polynomial-time algorithms using catalytic space that matches the $\tfrac{n}{2^{\Omega{\left( \sqrt{\log(n)} \right)}}}$ space-bound frontier for polynomial-time algorithms that solve the aforementioned problems, with all the space being borrowed from full memory except for $\bigO(\log n)$ bits. 

\vspace{0.1in}
We begin with $\stconn$, where additionally $\bigO(\log(n))$ random bits are sufficient to nearly match the time bounds of the best space-\emph{unrestricted} algorithms.
\begin{theorem} [\textbf{Sublinear catalytic space algorithms for $\stconn$}] \label{thm:main_stconn}
    There exists a deterministic algorithm solving $\stconn$ on $n$-vertex graphs in time $\poly(n)$, space $\bigO(\log(n))$, and catalytic space $\frac{n}{2^{\Omega{\left( \sqrt{\log(n)} \right)}}}$.
    
    Furthermore, for every $\epsilon > 0$, there is a randomized algorithm that runs in time $\bigO(|E| \cdot n^{2+\epsilon})$ given $\bigO(\log(n))$ random bits which uses the same space (both ordinary and catalytic) as the deterministic algorithm.
\end{theorem}
The randomized algorithm only makes \textit{one-sided errors}, where it always answers correctly when $s$ is not connected to $t$ in $G$. As such, the deterministic algorithm stated above is obtained by direct derandomization within the same space, i.e., by going over all possible random strings and accepting, if at least one of them makes the randomized algorithm accept. Our algorithms also improve on the $\Tilde{O}(n)$ (or the $\Tilde{O}(n^2)$) catalytic space bound observed by the randomized (or deterministic) algorithms from \cite{cook_pyne_26_efficient}, achieving a sublinear bound in both cases.

\vspace{0.1in}
Next, we extend this result to $\ed$, $\lcs$, and $\discreteFrechet$, by designing catalytic algorithms that exploit the special graph structure underlying the problems. To the best of our knowledge, no sublinear space algorithm was known for $\discreteFrechet$ even without the catalytic restriction.
\begin{theorem}[\textbf{Sub-linear catalytic space algorithms for $\ed, \lcs$, and $\discreteFrechet$}] \label{thm:main_other}
    There exist deterministic algorithms for $\ed$, $\lcs$, and $\discreteFrechet$ which use time $\poly(n)$, space $\bigO(\log(n))$, and catalytic space $\frac{n}{2^{\Omega{\left( \sqrt{\log(n)} \right)}}}$.
    
    Furthermore, the running time of our algorithm for $\discreteFrechet$ is $\bigO(n^{4+\varepsilon})$, for any $\epsilon > 0$, if all the distances of the points are integers from the range $\{0,\ldots,\poly(n)\}$.
\end{theorem}

We do not know what is the exact running time of our algorithms for edit distance and longest common subsequence as it depends on the running time of a log-space procedure\footnote{The running time for this procedure was not explicitly calculated and it boils down to the running time of a highly sophisticated log-space algorithm for multiplying $n$ integers of $\bigO(\log n)$ bits each.} for conversion between Chinese remainder representation and the usual binary representation of integers \cite{HesseAB02}.

To the best of our knowledge, no sublinear space algorithm was known for $\discreteFrechet$ even without the catalytic restriction. However, we note that such a theorem can be derived via the reduction to $\stconn$ and existing results: Chakraborty and Tewari~\cite{ChakrabortyT2017} give a polynomial time algorithm for $\stconn$ in directed layered planar graphs requiring $\bigO(n^\varepsilon)$ space which we can directly plug in to obtain a deterministic algorithm for $\discreteFrechet$ that uses time $\poly(n)$ and space $\bigO(n^\varepsilon)$.

This result is incomparable to our catalytic algorithm, as it uses less space overall but significantly more free memory.
We also note that the results on grid-graph reachability of Allender et al.~\cite{AllenderBCDR06} are not sufficient for our purposes as in the process of reducing $\discreteFrechet$ to grid reachability, some instances can have both multiple sources and multiple sinks.
Moreover, the grids we consider are NL-complete as opposed to the planar grids considered in~\cite{AllenderBCDR06} which are not known to be NL-complete.

Independently, Edenhofer~\cite{Edenhofer26} considered catalytic space-work space tradeoffs for $\stconn$ using techniques similar to ours and obtained incomparable results.

\subsection{Proof Overview}
\label{sec:intro:overview}

Our technique will primarily rely on the structure of two different $\stconn$ algorithms: first, the sublinear-space polynomial-time algorithm of Barnes et al.~\cite{BBRS98} in the non-catalytic world; and second, the recent polynomial-time algorithm of Cook and Pyne~\cite{cook_pyne_26_efficient} using linear catalytic space. We merge their approaches using analysis from other areas of catalytic computing, as well as incorporating new graph tools for overcoming one of the core barriers that arises in the approach of \cite{BBRS98}, to resolve the case of $\stconn$. Lastly, we use a common framework to attack $\ed$, $\lcs$, and $\discreteFrechet$, working on an implicit grid graph whose structure will allow us to focus on small subgraphs and thus keep our space small.

\subsubsection{Previous algorithms for $\stconn$: \cite{BBRS98} \& \cite{cook_pyne_26_efficient}}

\paragraph{Long and Short Paths Framework (\cite{BBRS98}).} The starting point for \cref{thm:main_stconn} is the deterministic algorithm by \cite{BBRS98}.
For any parameter $\lambda$, \cite{BBRS98} constructs a \textit{representative set} of vertices $U \subset V$ of size roughly $n/\lambda$ such that, if any pair $u$ is connected to $v$ in $G$ by a path $\rho_{u,v} = (u, p_1, \dots, p_L, v)$, then there exists a sequence of vertices $\Tilde{\rho}_{u,v} = (u, w_1, \dots, w_{L'},v)$, such that each $w_i$ is in $U$, and consecutive vertices in $\Tilde{\rho}_{u,v}$ are connected by a path of length at most $\lambda$. This provides a sparsification $H = (U,E_U)$ of $G$ that still preserves connectivity within $G$: $(s_i,s_j)$ is in $E_U$, if and only if there exists a path of length $\lambda$ from $s_i$ to $s_j$ in $G$. Thus by setting $\lambda$ to be $2^{\bigO{\left( \sqrt{\log(n)} \right)}}$, storing $U$ in the workspace suffices to decide $\stconn$ efficiently, if one can, in polynomial time and $n/\lambda$ workspace,
\begin{itemize}
    \item \textbf{construct} the representative set $U$;
    \item check if $(s_i,s_j)$ forms an edge in the sparsified graph $H$, through a ``\textbf{short path}'' algorithm that checks for connectivity over paths of length $\lambda$ in $G$; and
    \item using this short-paths algorithm as a sub-routine, find a ``\textbf{long path}'' in $H$, say $w_1, \dots, w_{L'}$, such that the distances from $s$ to $w_1$, and from $w_{L'}$ to $t$ are at most $\lambda$ in $G$.
\end{itemize} 

\noindent
To construct the set $U$, \cite{BBRS98} employs a breadth-first search (BFS) from $s$ to partition the vertices into $\lambda$ candidate sets $U_1, \dots, U_\lambda$, where intuitively each $U_i$ contains the vertices at distance $i, i+\lambda, \dots, i+ \lfloor n/\lambda \rfloor \cdot \lambda$ from $s$ in the BFS-tree. Each $U_i$ is a representative set by design, and by a simple averaging argument one of these sets, which we choose to be our true representative set, has at most $n/\lambda$ vertices. Thus if our ``short paths'' algorithm can test whether two nodes are at distance exactly $\lambda$, the vertices and edges of our sparsified graph $H$ can be computed in polynomial time and $|U| = \bigO(n/\lambda)$ space using BFS; this gives us our ``long paths'' algorithm.

To check if vertices $u$ and $v$ are connected by a path of length $\lambda$ in $G$, \cite{BBRS98} interpolate between ordinary BFS and the recursive approach by Savitch \cite{Sav70}. As a first step to reduce the space, they arbitrarily partition set $V$ into $\lambda$ color classes, each of size $n/\lambda$, and look only for $u$--$v$ paths which traverse over a fixed sequence of $\lambda+1$ (possibly repeatable) color classes, where the first class contains $u$ and the last class contains $v$. This is achieved by a basic BFS algorithm, iteratively computing connectivity between consecutive color classes in the sequence edge-by-edge and only storing the bit-vector of length $n/\lambda$ of which nodes in the current color class are reachable from $u$.

By looping over all possible sequences of color classes this discovers every possible path of length $\lambda$; however, the runtime of this procedure is at least $\lambda^{\lambda}$, which is superpolynomial when $\lambda$ is $2^{\sqrt{\log(n)}}$ as desired. Taking inspiration from Savitch's recursive algorithm, they instead choose a middle color class and test connectivity from the first color class, i.e. the one containing $u$, to this middle class and then from this middle class to the final class containing $v$. Such a test reduces connectivity of distance $\lambda$ to $2\lambda$ tests of distance $\lambda/2$, and recursively this gives a runtime of $(2\lambda)^{\log \lambda}$ instead, giving us our final choice of $\lambda = 2^{\bigO(\sqrt{\log(n)})}$.

\vspace{0.1in}
We aim to design a polynomial time algorithm for $\stconn$ which uses only $\bigO(\log(n))$ workspace and $\bigO(n/\lambda)$ catalytic space. The first issue with adapting their approach is that storing $U$ requires $\bigO(n/\lambda)$ clean workspace, which we cannot afford.
Furthermore, even if one generates $U$ in low-space, it is unclear how to construct catalytic algorithms for computing long and short paths that satisfy our time-space requirements.
With regards to the second issue, we now turn to the catalytic algorithm of Cook and Pyne~\cite{cook_pyne_26_efficient}, using $\bigO(\log(n))$ workspace and $\Tilde{\bigO}(n)$ catalytic space, which gives us an initial clue for our subroutines.

\paragraph*{Connectivity via Propagating Sums (\cite{cook_pyne_26_efficient}).}
In a word, the Cook-Pyne algorithm implements a catalytic version of BFS. To illustrate their main ideas, consider a layered graph $G$ with $n$ layers of $n$ vertices each, with edges only going between consecutive layers, and we want to test connectivity between $s$ in the first layer and $t$ in the last. We describe an algorithm using $\Tilde{\bigO}(n^2)$ catalytic space, which is linear in the number of vertices of $G$ (of course their algorithm works on arbitrary graphs, albeit with slightly less efficiency than in the layered case).

Let $R \coloneqq \Z/2^m\Z$ be the ring of integers modulo $2^m$, for some $m$ which we fix shortly. Their algorithm first allocates $n^2$ registers $r_{1,1} \ldots r_{n,n}$ on the catalytic tape, with each $r_{i,j}$ being assigned to corresponding node $v_{i,j}$ in $G$, and let their initial contents be $\tau_{1,1} \ldots \tau_{n,n} \in R$; this gives a total catalytic memory of $n^2 \cdot m$ bits. They start by ``pushing" the value of $r_{1,j}$ through all its outgoing edges to the second layer, in particular by adding $\tau_{1,j}$ to each node $v_{2,j'}$ connected to $v_{1,j}$. They then push the values of the $r_{2,j}$ nodes forward to the third level, where now each $r_{2,j}$ has not just $\tau_{2,j}$ but also any value added to it in the previous round.

Repeating this procedure for all layers $i$ in turn, 
we obtain some final value $\sigma$ which arrives at the vertex $(n,t)$ at the end of this process. They reset all catalytic memory by performing a ``reverse edge push'', i.e. repeating the whole procedure but with subtraction and in reverse, thus getting back $\tau_{i,j}$ in each register $r_{i,j}$. Finally, they repeat the entire algorithm but with the register $r_{1,s}$ for the source incremented by 1, collecting the value $\sigma'$ at the vertex $(n,t)$ and resetting as before.

For each pair of nodes $(v_{1,i}, v_{n,j})$ in the first and final layer respectively, let $\Gamma_{i,j}$ be the collection of paths between them. By linearity, the net result of our first forward push before resetting is
$$r_{n,j} \pluseq \sum_{i \in [n]} \sum_{P \in \Gamma_{i,j}} \sum_{v_{k,\ell} \in P} \tau_{k,\ell} \;\; (\mod 2^m)$$
because each value $\tau_{k,\ell}$ is added forward along the path $P$. In particular this holds for $t$, which gives us our $\sigma$ value, while the same argument in the second forward push gives
$$\sigma'= \left(\sum_{i \neq s} \sum_{P \in \Gamma_{i,t}} \sum_{v_{k,\ell} \in P} \tau_{k,\ell}\right) + \left(\sum_{P \in \Gamma_{s,t}} \left( 1 + \sum_{v_{k,\ell} \in P} \tau_{k,\ell} \right) \right) (\mod 2^m)$$
Hence $\sigma' - \sigma $ is exactly $|\Gamma_{s,t}| \pmod{2^m}$, i.e. the number of paths from $s$ to $t$ modulo $2^m$.

When $2^m > \sigma' - \sigma$ this resolves the connectivity between $s$ and $t$ and even accurately counts the exact number of $s$--$t$ paths.
However, the number of paths can be as large as $\bigO(n^n)$, and setting each register size to be of size $m = \Tilde{\bigO}(n)$ increases the catalytic space as well as the run-time by a factor of at least $n$. 
To deal with this, they finish the argument by using the fact that with large probability, a random prime $p$ of value $\poly(n)$ satisfies $\sigma' - \sigma \not\equiv 0 \pmod{p}$ if the number of paths is non-zero.
Hence, instead of computing over the ring $R$ they compute over the field $\F_p$, thus reducing the catalytic space to $\bigO(n^2 \cdot \log n)$ as required.

\vspace{0.1in}
Assuming one could stitch shorter segments together, we could use this algorithm for the long path procedure (putting aside the layered issue) on the graph $H$ of representative nodes.
However, using this for our short-paths sub-routine would still require the algorithm to push values across the edges associated with all the $n$ vertices, necessitating catalytic space of $\Tilde{\bigO}(n)$.

\subsubsection{Our approach for $\stconn$}
We now move on to our approach.
The two conceptual contributions of our work are the following, which we discuss in turn: 1) we adapt the flow-pushing argument of \cite{cook_pyne_26_efficient} to the recursive framework of \cite{BBRS98} to get the best of both worlds; and 2) we construct the representative set $U$ for the long paths algorithm using tools from the theory of pseudo-randomness in a space and time efficient manner.

\paragraph*{Merging Recursion and Flows.}
Assuming we have access to our representative set $U$, the long paths procedure can be handled by the \cite{cook_pyne_26_efficient} algorithm for unlayered graphs, and so we focus on finding paths of length $\lambda$ in $G$.
In combining the flow-based algorithm of \cite{cook_pyne_26_efficient} with the recursive approach of \cite{BBRS98}, we show both that the catalytic space of \cite{cook_pyne_26_efficient} can be reduced using techniques from \cite{BBRS98} and that the space of \cite{BBRS98} can be made catalytic using \cite{cook_pyne_26_efficient}.

If we want to adopt the recursive structure of \cite{BBRS98}, then we cannot use the idea of adding 1 to the initial vertex and comparing the sums before and after; this would require too much free memory and would make it unclear how to ``stitch together'' the two sides of a recursive call. Instead, we take inspiration from previous catalytic algorithms, such as that of Ben-Or and Cleve \cite{BC92}, which handle arithmetic cancellations in recursion. Our goal will be to count the number of paths between any pair of vertices $u$ and $v$.

For $\ell = 1 \ldots \log(\lambda)$, let $\Gamma_{u,v}^{\ell}$ be the set of $u$--$v$ paths of distance exactly $2^{\ell}$. \cite{BC92} give the following recursive invariant:
$$r_v \pluseq \tau_u \cdot |\Gamma_{u,v}^{\ell}|$$
Following \cite{BBRS98}, we may consider all possible midpoints for a $u$--$v$ path of length $2^{\ell}$.
Where before we simply took $\lor_w (u \rightarrow_{2^{\ell-1}} w) \land (w \rightarrow_{2^{\ell-1}} v)$, the crucial insight is that moving to arithmetic $\lor$ and $\land$, i.e. $+$ and $\times$, faithfully counts the number of paths in the same way, as any $u$--$w$ path can be paired with any $w$--$v$ path to give a $u$--$v$ path:
$$|\Gamma_{u,v}^{\ell}| = \sum_w |\Gamma_{u,w}^{\ell - 1}| \cdot |\Gamma_{w,v}^{\ell-1}|$$
Thus utilizing the \cite{BC92} procedure for recursive multiplication, fixing a midpoint $w$,
allows us to evaluate this inner product using four recursive calls:
\begin{enumerate}
    \item $r_v \minuseq r_w \cdot |\Gamma_{w,v}^{\ell-1}|$ \codecomment{$r_v = \tau_v - \tau_w \cdot |\Gamma_{w,v}^{\ell-1}|$}
    \item $r_w \pluseq r_u \cdot |\Gamma_{u,w}^{\ell-1}|$ \codecomment{$r_w = \tau_w + \tau_u \cdot |\Gamma_{u,w}^{\ell-1}|$}
    \item $r_v \pluseq r_w \cdot |\Gamma_{w,v}^{\ell-1}|$ \codecomment{$r_v = \tau_v + \tau_u \cdot |\Gamma_{u,w}^{\ell-1}| \cdot |\Gamma_{w,v}^{\ell-1}|$}
    \item $r_w \minuseq r_u \cdot |\Gamma_{u,w}^{\ell-1}|$ \codecomment{$r_w = \tau_w$ (and $r_u = \tau_u$)}
\end{enumerate}
Repeating this for all potential midpoints $w$ gives us the outer sum, and hence the invariant above. Repeating for every $u$--$v$ pair, by linearity we arrive at a final value of
$$r_v \pluseq \sum_u \tau_u \cdot |\Gamma_{u,v}^{\ell}| \quad \forall v \in V$$
which we take as the recursive invariant for our algorithm.
This fits the base case of $\ell = 0$, since we add $r_u$ to $r_v$ for each $u$ which has an edge to $v$, while the analysis of \cite{BC92} shows that it holds for $\ell$ given access to $\ell - 1$. Crucially for the catalytic restriction, at the end of each subroutine and hence the entire algorithm, $r_u$ and $r_w$ are left in their original configuration.

We now turn to the time of the recursion.
The above procedure, looping over all $w$ independently, is the structure of Savitch's algorithm~\cite{Sav70}, and so following \cite{BBRS98} we generalize our recursive calls to work over choosing middle color classes rather than individual vertices.
Once again we choose to have $\lambda$ color classes of size $n/\lambda$ each, and in making four recursive calls per middle color class (rather than the two from \cite{BBRS98}), our runtime becomes $(4\lambda)^{\log(\lambda)}$, which does not change our asymptotics for $\lambda = 2^{\sqrt{\log(n)}}$.
Our vectors of length $n/\lambda$ are entirely catalytic, with our free space being only the $\bigO(\log n)$ necessary to compute sums and keep track of our runtime.

As a final note, to make the analysis of our recursive condition simpler, we use the \emph{same} recursive algorithm for the long paths as well, with some minor interfacing between the two to reflect the change in graphs.
In the outer algorithm we run for paths of length up to $|U|$, i.e. maximal length paths, but in exchange we remove the color class structure and operate on the entire graph at every step. 
This gives a runtime of $4^{\log |U|} = \bigO(n^2)$, which is linearly less efficient than \cite{BBRS98} and \cite{cook_pyne_26_efficient}.

\paragraph*{Efficiently Constructible Representative Sets.}
\noindent We return to our final loose end for $\stconn$, which is the construction of the representative set for our long paths algorithm using low space. To do this, we use \emph{pairwise-independent hash function families} $\calH_{m,n}$ to construct a family of \textit{multisets} $\{U_\sigma\}$, indexed by $\sigma \in \{0,1\}^{(6+\eps)\log(n)}$, and prove that with high probability over the choice of $\sigma$, a multi-set from this family is a representative set for $G$. For convenience of exposition, we assume $\{U_\sigma\}$ as a family of sets, and deal with the technical challenges that arise in the $\stconn$ algorithm from $U_\sigma$ being a multi-set in \cref{sec:graph_decomp_bbrs}. 

Recall (from \cite{vadhan_psuedo_2012}) that for $m \leq n$, $\calH_{m,n}$ is a pair-wise independent family of hash functions $h : [m] \rightarrow [n]$ indexed by $2\log(n)$ bits, such that an $h$ chosen at random from $\calH_{m,n}$ maps any pair $i \neq j \in [m]$ into $[n]$ independently. Moreover, each hash function can be evaluated on an input $i \in [m]$, in $\polylog{n}$ time and $\bigO(\log(n))$ space (see \cref{lem:pwi_hash_props} for the formal statement).

For the construction, we set $m = \bigO(n/\lambda)$, and suppose that we sample $K = \bigO(\log(n))$ many hash functions $h_1, \dots, h_K$ from $\calH_{m,n}$, and consider $U_\sigma$ as the union of their images, where $\sigma$ is the random seed used to sample them. Using a probabilistic argument in \cref{sec:hitting_sets}, we show that with high probability over the choice of the $h_1, \dots, h_K$ (or $\sigma$) the following holds: for every pair of vertices $u,v$ from $G$ that are connected by a path of length $\lambda/2$, there exists a walk between them traversing $\lambda/2-1$ distinct vertices that visits $U_\sigma$. 

With high probability over $\sigma$, $U_\sigma$ forms a representative set for $G$, over paths of length $\lambda$. Indeed, for any $u,v$ that is connected, suppose we break this path into $L'$ sub-paths of size $\lambda/2$. By the properties of a ``good" $U_\sigma$, we see that each sub-path can be replaced by a walk traversing $\lambda/2-1$ distinct vertices and visiting $U_\sigma$. 
The distance between consecutive vertices from $U_\sigma$ along this new path from $u$ to $v$ is at most $\lambda$, so $U_\sigma$ can be used to sparsify $G$.

However, the length of $\sigma$ is $2K\log(n) = \bigO(\log^2(n))$, and we cannot store this within our $\bigO(\log(n))$-sized workspace. As such, we use the well-known technique of reducing randomness by performing a random walk on a \textit{constant-degree explicit expander} $D$, with suitably high expansion \cite{AKS87_derand}. In particular, sampling $K$ hash functions using a random walk on $D$ specified over vertex set $\{0,1\}^{2\log(n)}$ (that describes all of $\calH_{m,n}$), gives similar guarantees of intersecting with walks, as that of picking $K$ hash functions uniformly at random. The number of random bits is now at most $\bigO(\log(n))$, and by carefully accounting the values of $K$, the degree and expansion of $D$, a seed length of $(6+\eps)\log(n)$ suffices to describe $U_\sigma$.

Finally, any vertex in $U_\sigma$ can be generated, given $\sigma$, in $\polylog{n}$ time and $\bigO(\log(n))$ space. This holds from standard explicit constructions of constant-degree expanders, and efficient evaluation of hash functions in $\calH_{m,n}$.

\subsubsection{$\ed$, $\lcs$, \& $\discreteFrechet$ Without Direct Reductions To $\stconn$}
We now turn our attention to computing $\ed$, $\lcs$, \& $\discreteFrechet$.
A natural direction would be to use the reduction to reachability of grid-like graphs and use the previous algorithm.
However, on inputs of length $n$, the underlying grid-like graph has quadratic size, and thus the space required by the previous algorithm would be $\frac{n^2}{2^{\Omega(\sqrt{\log n})}}$ -- far larger than the space we would like to use.

\paragraph*{$\lcs$ and $\ed$ in Sublinear Free Space (\cite{KiyomiHO21}).}
While~\cite{BBRS98} have the outer algorithm with a single color class (corresponding to the representative set) and recursion depth $\bigO(\log n)$ and the inner algorithm with $\lambda$ color classes and recursion depth $\bigO(\log(\lambda))=\bigO(\sqrt{\log n})$, \cite{KiyomiHO21} approach the problem by effectively using the ideas similar to the inner algorithm of~\cite{BBRS98} as their outer algorithm.
In the outer algorithm, given a color class in the first layer\footnote{The grid-like graph is not layered in the usual sense, which is not an issue in this algorithm as it only uses free space.} and a color class in the last layer, \cite{KiyomiHO21} choose a middle layer of the grid-like graph and then iterate over the $\lambda$ color classes of size $n/\lambda$ as the possible choices of a middle vertex of the path, doing this recursively until a base case is reached, when the two layers are at distance $n/\lambda$.
In this base case, the inner algorithm takes over.
Because the grid-like graph only has edges to close neighbors, either the two color classes are far enough so that there is no possible path between the two, or the graph containing the two color classes and all paths between them contains $\bigO(n/\lambda)$ rows and $\bigO(n/\lambda)$ columns, and therefore we can compute the shortest path from the first color class to the second color class by the standard dynamic programming algorithm that only requires us to remember two rows at a time.

It is not immediate that one can emulate the algorithm of  \cite{KiyomiHO21} using catalytic tools.
In particular, their algorithm is designed to find the shortest path which is necessary for computing edit distance.
Thus the algorithm always has to choose the shortest path that leads to a given middle point.
However, the ``push" algorithm of Cook and Pyne sums all the paths leading to the middle point.
It is not clear how to select the minimum path among them.
More to the point, the values we are summing in the algorithm of Cook and Pyne have nothing to do with the true values we want to take minimum of
as they are offset by some unknown initial values from the catalytic memory.
Taking their minimum would be completely useless.

The standard approach to overcome this would be to layer the graph and count the length of the paths implicitly.
Checking whether the length of a path from a source vertex to the target vertex is of length at most $d$ is done 
by checking whether the source vertex in the first layer is connected to the target vertex in the $d$-th layer.
However, the layering increases the size of the graph by a factor of $\bigO(n)$ so any possible savings in space from clever algorithms will be wiped-out.

To overcome this, we stick to the original graph but introduce multiplicative weights to the edges.
It turns out the the algebraic framework of the Cook-Pyne algorithm can be modified to count a {\em weighted sum} of the paths.
The weight of a path will be the product of its edge weights.
By choosing weights that are large enough, different path weights will be separated so that from their weighted sum we can read off the count of the paths of different weights.
This assumes we will work with large integers that we will not have space to store;
the path weights will have $\bigO(n^2)$ bits.
We overcome this by computing the final value modulo different primes, that is we compute it in Chinese remainder representation.
Using a space-efficient algorithm that converts Chinese remainder representation into the usual binary representation,
we get bit-by-bit access to the actual weighted sum.
From it we can read off the length of the shortest path as well as the length of the longest path.

So our goal will be to design an algorithm for weighted grid graphs that computes the weighted sum of paths between two vertices modulo a small prime.

\paragraph*{Weighted Reachability on Grid Graphs.}
We use a similar high-level technique to finding the shortest path in grid-like graphs as the algorithm~\cite{KiyomiHO21}.
Instead of finding the shortest path we will  compute the weighted sum of the paths.

Again, our algorithm for weighted reachability on grid graphs will be composed of an outer and an inner algorithm.
For practical purposes, we assume that the edge weights are given to us by an oracle, and we use this oracle as a parameter for our three functions of interest which we want to compute.
Moreover, for this algorithm we assume all values to be over a field $\F=\F_p$ for some prime $p=\poly(n)$.

The outer algorithm will be essentially identical to the inner ``short paths'' algorithm from our algorithm for general $\stconn$.
It will use recursion of depth $\bigO(\sqrt{\log n})$ with color classes of size $n/\lambda$, for $\lambda = 2^{\Theta(\sqrt{\log n})}$,
to reduce the weighted reachability on a grid of size $n\times n$ to the weighted reachability on grid-graphs of size roughly $n/\lambda \times n/\lambda$.
This smaller reachability will be solved by the inner algorithm.

However, we do not have the luxury of using our previous algorithm as our inner algorithm, given that our goal is to use only $\bigO(\log n)$ free space.
Therefore, we use a similar idea to the outer algorithm of~\cite{BBRS98} together with the ideas of Cook and Pyne.
We are interested in computing weighted reachability for two color classes of size $n/\lambda$ at horizontal distance $n/\lambda$.
Two color classes at such horizontal distance in a grid-graph can affect each other only if they are adjacent vertically. 
Thus, we get a subgrid with $n/\lambda$ layers and at most three color classes, yielding $3\cdot n/\lambda$ vertices per layer.
We now use the idea of propagating sums for the outer algorithm in our $\stconn$ algorithm while recursing again on the middle layers.
We will now have only a single color class, which is not an issue with $\bigO(n/\lambda)$ vertices per layer.
Our final base case here is the case of two neighboring layers, where we can iterate over all $\bigO(n/\lambda)$ edges and multiply intermediate values by their weight.

Here, we only need a constant number of bits of free space per level of recursion for the instruction counter, and therefore the total recursion depth $\bigO(\log n)$ does not pose any large issues space-wise.
Moreover, we can implement the recursion in a constant number of steps per each recursion level, thus we get time complexity $c^{\bigO(\log n)}\cdot \bigO(n/\lambda)=\poly(n)$.
Regarding catalytic space usage, if we use a new catalytic register vector at every recursion level, we require $\bigO(\log n)$ vectors, each using $n/\lambda$ field elements, which each requires $\bigO(\log p)$ bits, which in total yields the required catalytic space complexity $\frac{n}{2^{\Omega(\sqrt{\log n})}}$.

To finally compute the weight of all $u,v$-paths modulo $p$, we first run the algorithm forward, recording the result $r_1$, then rerun the computation backwards.
We than add 1 to the register corresponding to $u$, and run the computation again, recording the result $r_2$, and rerun the computation backwards for the final time.
We then subtract 1 from the register corresponding to $u$ and return $r_2-r_1$ as the result.

\paragraph*{Using Grid Reachability to Compute Discrete Fr\'{e}chet Distance.}
For discrete Fr\'{e}chet distance of point sequences $P=(p_1,\ldots,p_n)$ and $Q=(q_1,\ldots,q_n)$, we use the previous algorithm as a blackbox, and we provide a suitable edge weight oracle.
We also assume that the distances of the points are integers in the range $\{0,\ldots,\poly(n)\}$.
If we are given a more complicated metric, we can build a range query oracle out of it -- that is, we can compute how many different point pairs have distance lower than the queried point pair $(p,q)$, by iterating over all the other pairs, and thus we would get an oracle whose values are in the range $\{1,\ldots,n^2\}$.

One possible way of computing discrete Fr\'{e}chet distance is by taking a grid-like graph on vertex set $\{0,\ldots,n\}\times\{0,\ldots,n\}$, where each vertex $(i,j)$ has three outgoing edges into vertices $(i+1,j), (i,j+1), (i+1,j+1)$ if these vertices exist.
We then build an induced subgraph which we call $FRECH_{P,Q,m}$, where we set the weight of each edge $((i,j), (k,\ell))$ to one if the distance between $p_k$ and $q_\ell$ is at most $m$ and zero otherwise.
As there are only $\poly(n)$ possible values for $m$, we can try all possible values of $m$, and the smallest value of $m$ such that there is a path between $(0,0)$ and $(n,n)$ is the Fr\'{e}chet distance.
However, there may be exponentially many paths and we only allow computation modulo primes of polynomial size.
We remedy this using the Chinese remainder theorem and the fact that the total weight is zero iff the total weight module the first $4n$ primes is zero.
Thus, we run the algorithm for $4n$ primes with the same weights (as they are all either zero or one) to determine whether the discrete Fr\'{e}chet distance is at most $m$, and we use this to run binary search over all choices of $m$ to compute $\discreteFrechet$ in polynomial time, logarithmic free space and catalytic space $\frac{n}{2^{\Omega(\sqrt{\log n})}}$.

\paragraph*{Using Grid Reachability to Compute $\ed$ and $\lcs$.}
For $\ed$ and $\lcs$ for $n$-character strings, we want to find the shortest or, respectively, longest path in a weighted grid-like graph -- essentially the graph for $\discreteFrechet$. However, unlike $\discreteFrechet$, we need to take care of different weights of edges.
Furthermore, in our approach, we multiply weights instead of adding them, and hence we must move to addition in the exponent.

For the grid-like graph, we note that there at most $3^{2n}$ paths (the longest path has length $2n$, and each vertex has out-degree at most three).
In all cases, the edges in the grid-like graph have two possible weights: 0 or 1, but some edges might also be missing.
Therefore, we can simulate addition by having the weights as follows: if the edge is missing, its weight will be set to 0.
(And therefore, any path using such edge will have zero total weight, which effectively removes the edge.)
For the existing edges, if an edge had weight $\mu(e)$ originally, then our new weight will be $16^{n\cdot \mu(e)}$.
Therefore, weight-zero edges will have their new multiplicative weight set to 1, and weight-one edges will have their multiplicative weight set to $16^n$.
This also ensures that we can efficiently learn the longest/shortest path in the graph: given the result, which is a sum of weights of all $s,t$-paths, the lowest bit set to 1 corresponds to the length of the shortest path: indexing bits such that the least significant bit has index 0, any of the bits $4n(w+1)-1,\ldots,4nw$ being set to one corresponds to there existing a path of weight $w$ in the original additive weighting scheme.
Therefore, we can just check all the bits of the number to find the largest and the smallest non-zero bit index.

There is one slight issue: if we were computing this over integers, the representation of the value $r$ would have $\Theta(n^2)$ bits, which we cannot afford to represent directly.
However, we can use the Hesse-Allender-Barrington algorithm~\cite{HesseAB02} for reconstruction integers from their Chinese remainder representation that uses polynomial time, polynomial number of queries for the value of the $r\bmod p_i$ for some primes $p_i$ and logarithmic space to output a particular bit of $r$ in the binary representation.
And since our algorithm returns all results modulo $p$, each query to $r\bmod p_i$ corresponds to us running our reachability algorithm over the field $\F_{p_i}$.
Thus, we get again polynomial time, logarithmic free space and catalytic space $\frac{n}{2^{\Omega(\sqrt{\log n})}}$.

We note that this approach also works for the weighted versions of edit distance if the weights are $\bigO(\log n)$-bit integers with appropriate changes to the weighting scheme.

\subsection{Future Work}
\label{sec:intro:open}

We consider these algorithms to be an initial step towards bypassing the current $\frac{n}{2^{\Theta\left( \sqrt{\log(n)} \right)}}$ space barrier for these problems. The NNJAG lower bounds which substantiate this barrier are based on pebbling techniques, which was also the case for Tree Evaluation for over a decade before finally being broken by the catalytic algorithm of \cite{CM24_tree_eval}. Their algorithm bypasses these pebbling lower bounds by combining multiple pebbles in the same location, a process which has been called ``fuzzy pebbling'' by later reviews on the topic. This is also the situation for our $\stconn$ algorithm, where the connectivity of many nodes in the graph are all stored on top of one another; thus there is hope that further development of catalytic techniques or more clever instantiations will allow us to go beyond this longstanding space-time barrier.
\section{Preliminaries}
\label{sec:prelims}

\subsection{Catalytic Model}
Our machine model in this paper is an extension of ordinary space-bounded computation:

\begin{definition}[Catalytic Turing machine]
    A \emph{catalytic Turing machine} is a Turing machine with four tapes: 1) a read-only \emph{input} tape; 2) a write-only \emph{output} tape; 3) a read-write \emph{work} tape; and 4) a read-write \emph{catalytic} tape.
    
    We say that a function $f: \{0,1\}^n \rightarrow \{0,1\}^m$ can be solved in $\CSPACE[s,c]$ if there exists a catalytic Turing machine $M$ whose input, output, work, and catalytic tapes have length $n$, $m$, $s$, and $c$ respectively, such that for any $x \in \{0,1\}^n$ and $\tau \in \{0,1\}^c$, if $M$ is initialized with $x$ on the input tape and $\tau$ on the catalytic tape, then after execution $M$ halts with $f(x)$ on the output tape and $\tau$ on the catalytic tape.

    We further say that $f$ is in $\CTS[t,s,c]$ if $M$ runs in time $t$ for every $x$ and $\tau$.
\end{definition}

Since time is a consideration in our algorithm, we specify that in this paper we will use {\em Random Access Turing machines}, which allow us to jump with their heads to a particular tape cell, specified on an auxiliary {\em address} tape consisting of $\lceil \log(c+s)\rceil$ bits, as an atomic operation. This does not provide the machines with extra computational power, but only allows for faster simulation of RAM machines.

While our ultimate goal in catalytic computation is to compute a function into ordinary memory while leaving catalytic memory untouched, along the way it is necessary to alter the catalytic tape in predictable (and reversible) ways:

\begin{definition}[Catalytic subroutine]
    A \emph{catalytic subroutine} is a program $P$ whose net result is to 1) transform a section of the catalytic tape $\tau$ via addition, i.e. $\tau \pluseq v$ for some value $v$; 2) manipulate the free memory in any way; and 3) leave all other catalytic memory in its original configuration. We also require that every catalytic subroutine $P$ has an inverse program $P^{-1}$ whose result is to subtract $v$ from the same region instead.
\end{definition}

Our algorithms will need to operate with catalytic memory containing values from a particular field $\F_p$. In the appendix we explain how this can be done efficiently.

\subsection{Function Statements}
We first define all functions of interest in this paper:

\begin{definition}[Connectivity]
    For a directed graph $G = (V,E)$ and two vertices $s,t \in V$, we say $s$ and $t$ are \emph{connected} if there is a path from $s$ to $t$ in $G$.
\end{definition}

\begin{definition}[Longest common subsequence]
    For a string $x\in \Sigma^*$, a \emph{subsequence} is a string $z$ that can be obtained by removing any number of characters of $x$.

    For two strings $x,y\in\Sigma^*$, the \emph{longest common subsequence} is the longest $z\in\Sigma^*$ that is a subsequence of both $x$ and $y$, and we define the corresponding function $\lcs(x,y):=|z|$.
\end{definition}
\begin{definition}[Edit distance]
    Given a string $x\in \Sigma^*$, an \emph{edit operation} is either a deletion, insertion or a substitution of a character in $x$.

    For two strings $x,y\in \Sigma^*$, their \emph{edit distance} $\ed(x,y)$ is the minimum number of edit operations needed to transform $x$ into $y$.
\end{definition}

\begin{definition}[Discrete Fr\'{e}chet distance]
    For two sequences $P=(p_1,\ldots,p_m), Q=(q_1,\ldots,q_n)$ of points from a metric space $(M,d)$, their \emph{monotone alignment} is a sequence of pairs $(i_k,j_k)\in [m]\times [n]$ of some length $\ell$ such that $(i_1,j_1)=(1,1)$, $(i_\ell,j_\ell)=(m,n)$ and for all $2\leq k\leq \ell$, $i_k \in \{ i_{k-1}, i_{k-1}+1\}$ and $j_k \in \{ j_{k-1}, j_{k-1}+1\}$.
    
    The \emph{discrete Fr\'{e}chet distance} of $P$ and $Q$ is 
    \[
        \discreteFrechet(P,Q)=\min_{A}\max_{(i,j)\in A} d(p_i,q_j),
    \] where $A$ ranges over all possible monotone alignments.
\end{definition}
Informally, the discrete Fr\'{e}chet distance describes the situation, where there are two frogs connected by a rope at points $p_1,q_1$, and in each step, each frog may jump forward to the next point in the sequence.
The distance itself is then the minimum length of the rope such that the frogs can jump in a way that the rope does not snap.

In our algorithm, we assume that the metric $d$ is represented by integers from the range $\{0,\ldots,\poly(n)\}$.
This assumption is without loss of generality: there are $m\cdot n$ pairs in $P\times Q$, and it follows from the definition that the distance of one of the pairs is the discrete Fr\'{e}chet distance, and our algorithm will only use the queries of the form $(p_i,q_j)$ for some $p_i\in P, q_j\in Q$.
Thus, we can represent the order of distances in the metric by simply ranking the distances between the pairs of points.
These results are integers from the set $\{1,\ldots,m\cdot n\}$, which is sufficient.
Provided the distances from the original metric can be compared using $\bigO(\log n)$ bits of space, we can translate from the metric into the ranks in logarithmic space and roughly quadratic time per query.

\subsection{Pairwise Independent Hash Functions}
Our algorithms for $\stconn$ will require efficiently constructible \emph{hash functions} for sparsifying our graph $G$:

\begin{definition}[Pairwise Independent Hash Functions]
    \label{def:pwi_hash}
    A family of functions $\calH_{m,n} = \{h : [m] \rightarrow [n] \}$ is \textsf{pairwise independent} if for every $x_1 \neq x_2 \in [m]$, and every $y_1, y_2 \in [n]$, when a function $h$ is chosen uniformly at random from $\calH_{m,n}$,
    $$\Pr_{h \sim \calH_{m,n}} [h(x_1) = y_1 \land h(x_2) = y_2] =  \frac{1}{n^2}.$$
    In addition, for any $h \in \calH_{m,n}$, let $\mathsf{Im}(h)$ be the image of $h$.
\end{definition}
\noindent It follows from \Cref{def:pwi_hash} that for every $x, y \in [m]$, $\displaystyle \Pr_{h \sim \calH_{m,n}}[h(x) = y] = 1/n$. 

Important for our result is that there exist pairwise independent hash functions that have a short description and can be constructed efficiently:
\begin{lemma}[\cite{vadhan_psuedo_2012}]
    \label{lem:pwi_hash_props}
    For every $m,n \in \N$, there exists a pairwise independent hash function family $\calH_{m,n}$, where a random function can be sampled from $\calH_{m,n}$ using $r = \max \{\log(m),\log(n)\} + \log(n)$ bits. 
    
    Moreover, given any $h \in \calH_{m,n}$ (described by a string in $\Bits^r$) and $x \in [m]$, $h(x)$ can be computed in time $\polylog{m,n}$ and space $\bigO(r)$.
\end{lemma}

\subsection{Modular Arithmetic}

In addition to hash functions, we will also require working over the integers modulo a prime $p$ of logarithmic bit-size in an efficient manner:

\begin{proposition} \label{prop:fieldwork}
    Given a prime $p$, we can compute addition, subtraction, and multiplication modulo $p$, in time $\polylog{p}$ and space $\bigO(\log(p))$.
\end{proposition}

\noindent
Working with modulo primes will be a useful proxy for testing whether or not a number is non-zero:

\begin{proposition} \label{prop:rand_prime}
    Let $0 < m \leq 2^t$, and let $t \log t \leq p \leq 2t \log t$ be a randomly chosen prime. Then $m \not\equiv 0 \mod p$ with probability 0.99 over $p$.
\end{proposition}

\noindent
Beyond testing for zero, we will sometimes need to represent numbers exactly, for which we use the \emph{Chinese remainder representation} of a number $x$ over many primes.
There is a log-space procedure that can calculate bits of $x$ in binary representation using {\em oracle} access to the Chinese remainder representation of $x$ \cite{HesseAB02}:
\begin{proposition}[\cite{HesseAB02}]\label{prop-CRR}
Let $p_i$ denote the $i$-th prime.
There is a procedure that works in space $\bigO(\log n)$ with query access to $x \bmod p_i$, for $i=1,\dots, n$, where $0\le x <2^n$ is an arbitrary integer, 
that on input $(n,j)$ represented in binary, where $n>j\ge 0$ are integers, outputs $j$-th bit of the binary representation of $x$.
The procedure runs in polynomial time and makes in total $\bigO(n^2)$ queries to $x \bmod p_i$.
\end{proposition}

\noindent
We remark that the procedure would work as long as $x< \prod_{i=1}^n p_i$ which would provide $\bigO(\log n)$-factor savings in time and the number of queries.
This is however irrelevant for us.
Moreover the procedure can output all the bits of the binary representation of $x$ in the usual order using the same space, time and number of queries.
This could provide time savings of $\poly(n)$-factor for our edit distance and LCS algorithms. 

\section{Pushing Flow Through Walks}
\label{sec:push}

In this section we will utilize the short paths (or in our case, walks) decomposition from \cite{BBRS98} as a generic setup for the flow-based algorithm of \cite{cook_pyne_26_efficient}. The latter assigned initial values to every node in the graph from the catalytic tape, and then ``pushed'' these values forward through the graph along every path via addition, using cancellation to collect up only terms from relevant (i.e., $s$--$t$) paths. To improve their catalytic space usage, we need to follow the color class framework of~\cite{BBRS98} and use recursion restricted to subsections of the graph.

\begin{definition}[Walk matrix]\label{def:path_flow}
    Let $H = (V,E)$ be a graph on $m \in \mathbb{N}$ vertices, and for convenience assume $m$ is a power of 2. Fix a field $\F = \mathbb{F}_p$ for some prime $p \leq \poly(m)$, and let $\mu: E \rightarrow \F$ be a weighting function on the edges of $H$. Given a walk $W$ in $H$, we say the weight of $W$ is the product of all edge weights in $W$ with multiplicity, i.e. $\mu(W) = \prod_{e \in W} \mu(e)^{d_e}$ where $d_e$ is the number of times $e$ occurs in $W$.
    
    For any $\ell \in [\log m]$, and any $\vin,\vout \in V$, define $\Gamma_{\ell}(\vin,\vout)$ to be the collection of all walks of length exactly $2^{\ell}$ from $\vin$ to $\vout$ in $H$. We define the \emph{walk matrix} $\calW_{\ell} \in \F^{V \times V}$ by
    $$\calW_{\ell}[\vin,\vout] = \sum_{W \in \Gamma_{\ell}(\vin,\vout)} \mu(W)$$
\end{definition}

While we cannot store the walk matrix, we simplify our task in two ways: 1) we focus on a collection of input and output nodes, $\bigvin$ and $\bigvout$ respectively, of sufficiently small size to be stored; and 2) while we still cannot store the matrix $\calW_{\ell}$ restricted to $\bigvin \times \bigvout$, we instead provide an input vector $\tauin \in \F^{\bigvin}$ and compute the output of $\tauin \cdot \calW_{\ell} $ in the entries $\bigvout$.

\begin{remark}
    For convenience, given sets $\bigvin, \bigvout \subseteq V$ and a vector $\tauin \in \F^{\bigvin}$, we write $\tauin \cdot \calW_{\ell} $ to mean the vector resulting from the product of matrix $\calW_{\ell}$ with the (row) vector indexed by $V$ which is $\tauin$ for all $\vin \in \bigvin$  and 0 everywhere else. If we add this to a vector $\tauout \in \F^{\bigvout}$, we mean that we restrict this vector to the entries indexed by $\vout \in \bigvout$ and do addition entry-wise by corresponding $\vout$.\footnote{Note that this definition is \textit{not} the same as taking the product with submatrix $(\calW_{\ell})_{\bigvin, \bigvout}$; we will allow our walks to go via any vertices as long as they start in $\bigvin$ and end in $\bigvout$.} Thus we have
    $$(\tauin \cdot \calW_{\ell})[\vout] = \sum_{\vin \in \bigvin} \tauin[\vin] \cdot \sum_{W \in \Gamma_{\ell}(\vin,\vout)} \mu(W)$$
\end{remark}

Our $\bigvin$ and $\bigvout$ sets will be chosen via breaking $V$ into easily computable blocks of equal size, often referred to as color classes.

\begin{definition}\label{def:colors}
    Let $H$ be a graph on $m \in \mathbb{N}$ vertices, and let $C \in [m]$ such that $C$ and $m$ are powers of 2. We partition $V$ into sets $\{V_c\}_{c \in [C]}$, which we call \emph{color classes}, such that $|V_c| = m/C$ for all $c$, by letting the first $\log C$ bits of any vertex $v$'s label in $[m]$ indicate the color class of $v$.
\end{definition}
While the key property for our space and time bounds is the even distribution of nodes into classes, note that computing the color class and label within the color class of $v$ can be done in time $\log C$ and $\log m/C$ respectively.

We now reach our main propagation subroutine which will drive all algorithms in this paper. Unlike \cite{BBRS98} itself, we will require calculating each side of the recursion multiple times in order to get the cancellations necessary for \cite{cook_pyne_26_efficient}; however, by setting the parameters correctly we maintain an equivalent asymptotic runtime.

\begin{lemma} \label{lem:main_subroutine}
    Let $m \in \mathbb{N}$, let $\F = \mathbb{F}_p$ for some prime $p = \poly(m)$, and let $H = (V,E)$ be a graph on $m$ vertices with edge weight function $\mu$. Let $\calW$ be defined with respect to $H$ as per \Cref{def:path_flow} and let $C$ and $\{V_i\}_{i \in [C]}$ be as per \Cref{def:colors}.
    
    Assume that we have a catalytic subroutine $P_0$ which takes in $\cin, \cout \in [C]$ as well as sections $\bigRin$, $\bigRout$, $\bigRmid \in \F^{m/C}$ from the catalytic tape, and adds $\bigRin \cdot \calW_0 $ to $\bigRout$.
    
    Then there exists a catalytic subroutine $P$ which takes in $\ell \in [\log m]$, $\cin, \cout \in [C]$, and sections $\bigRin, \bigRout, \bigRmid \in \F^{m/C}$ from the catalytic tape, and whose effect is to add $\bigRin \cdot \calW_{\ell}$ to $\bigRout$.
    
    If $P_0$ runs in time $T$ and uses free space $S$, then $P$ runs in time $(4C)^{\ell} \cdot T$ and uses free space $\ell \log 4C + S$; it uses only the catalytic memory $\bigRin, \bigRout, \bigRmid$ and the memory required to run $P_0$.
\end{lemma}

\begin{proof}
    For each $\ell$, we will construct a program $P_{\ell}$ which is equivalent to $P$ with input $\ell$; we do this inductively on $\ell$ from 0 up to the maximum value of $\log m$. The base case is taken care of by program $P_0$, which runs in time $T = (4C)^0 \cdot T$ as expected.

    Now consider the case of $\ell \geq 1$. Fix some $\cmid \in [C]$, and let $[mid \rightarrow out]^+$ denote the program $P_{\ell-1}(\cmid,\cout;\bigRmid,\bigRout,\bigRin)$, let $[mid \rightarrow out]^-$ denote its inverse,\footnote{Recall that $P_0$ has an inverse program by assumption of it being a catalytic subroutine, while $(P_{\ell})^{-1}$ can be obtained by running the above program backwards and inverting each line.} and let $[in \rightarrow mid]^+$ and $[in \rightarrow mid]^-$ be $P_{\ell-1}(\cin,\cmid;\bigRin,\bigRmid,\bigRout)$ and its inverse as well. Our program $P_{\ell}$ is as follows:
    \begin{itemize}
        \item for all $\cmid \in [C]$:
        \begin{enumerate}
            \item $[mid \rightarrow out]^-$
            \item $[in \rightarrow mid]^+$
            \item $[mid \rightarrow out]^+$
            \item $[in \rightarrow mid]^-$
        \end{enumerate}
    \end{itemize}
    First we show this program fits our time and space restrictions. By induction, each loop iteration takes time $4C \cdot ((4C)^{\ell - 1} \cdot T)$, which gives time $(4C)^{\ell} \cdot T$ in total. We use no additional catalytic memory and only require $(\ell - 1)(\log 4C) + S$ plus an additional $\log 4C$ bits of memory to track $\cmid$ and the inner instruction count, for $\ell \log 4C + S$ free space in total.\footnote{The inputs to the current recursive program, i.e. the values of $\ell$, $\cin$, and $\cout$, as well as the current addressing of $\bigRin$, $\bigRout$, and $\bigRmid$, are all tracked implicitly as follows. Since $P$ has access to $\ell$, and we are currently storing a stack of $\ell' \leq \ell$ strings of length $\log 4C$, we know that we are currently running $P_{\ell-\ell'}$ and the most recent string gives us $\cin$ and $\cout$. Furthermore, the extra two bits from each level tell us our trace of which call we are on at each level, which determine some rotation of each of the three registers.}
    
    We now prove the correctness of our program. Note that for any $(\vin, \vout) \in V_{\cin} \times V_{\cout}$, every walk of length exactly $2^{\ell}$ must pass through some midpoint $\vmid$ at distance $2^{\ell-1}$ from both; furthermore, any walk from $\vin$ to $\vmid$ may be paired with any walk from $\vmid$ to $\vout$ to form a distinct walk from $\vin$ to $\vout$. Thus we have the following decomposition:
    $$\Gamma_{\ell}(\vin,\vout) = \displaystyle\bigcup_{\cmid \in [C]} \left[\displaystyle\bigcup_{\vmid \in V_{\cmid}} \Gamma_{\ell-1}(\vin,\vmid) \otimes \Gamma_{\ell-1}(\vmid,\vout)\right]$$
    where both unions are disjoint. By extension, for any $(\vin, \vout) \in V_{\cin} \times V_{\cout}$,
    $$\sum_{W \in \Gamma_{\ell}(\vin,\vout)}\mu(W) = \sum_{\cmid \in [C]} \left[\sum_{\vmid \in V_{\cmid}}\left(\sum_{W \in \Gamma_{\ell-1}(\vin,\vmid)}\mu(W)\right) \cdot \left(\sum_{W \in \Gamma_{\ell-1}(\vmid,\vout)}\mu(W)\right)\right]$$
    Given this fact, we can analyze our algorithm for the first value of $\cmid \in [C]$. Let $\tauin$, $\tauout$, and $\taumid$ be the initial values in $\bigRin$, $\bigRout$, and $\bigRmid$ respectively. After our first instruction, $\bigRout$ has values
    $$\bigRout[\vout] = \tauout[\vout] - \left[\sum_{\vmid \in V_{\cmid}} \taumid[\vmid] \cdot \left(\sum_{W \in \Gamma_{\ell-1}(\vmid,\vout)}\mu(W)\right)\right]$$
    After our second instruction, $\bigRmid$ has values
    $$\bigRmid[\vmid] = \taumid[\vmid] + \left[\sum_{\vin \in V_{\cin}} \tauin[\vin] \cdot \left( \sum_{W \in \Gamma_{\ell-1}(\vin,\vmid)}\mu(W)\right)\right]$$
    Thus after our third instruction, $\bigRout$ has values
    \begin{equation*}
        \begin{split}
            \bigRout[\vout] &= \tauout[\vout] - \left[\sum_{\vmid \in V_{\cmid}} \taumid[\vmid] \cdot \left(\sum_{W \in \Gamma_{\ell-1}(\vmid,\vout)}\mu(W)\right)\right] \\
            &+ \sum_{\vmid \in V_{\cmid}} \left[ \taumid[\vmid] + \sum_{\vin \in V_{\cin}} \tauin[\vin] 
            \cdot \left(\sum_{W \in \Gamma_{\ell-1}(\vin,\vmid)}\mu(W)\right) \right]
            \cdot \left(\sum_{W \in \Gamma_{\ell-1}(\vmid,\vout)}\mu(W) \right) \\
            &= \tauout[\vout] + \displaystyle\sum_{\vmid \in V_{\cmid}} \sum_{\vin \in V_{\cin}} \tauin[\vin] \cdot
            \left(\sum_{W \in \Gamma_{\ell-1}(\vin,\vmid)}\mu(W)\right) \cdot \left(\sum_{W \in \Gamma_{\ell-1}(\vmid,\vout)}\mu(W)\right) \\
        \end{split}
    \end{equation*}
    Our fourth instruction is the reverse of the second, and so we are left with $\bigRmid$ in its initial state $\taumid$, while $\bigRin$ was never altered; as required, this guarantees that only $\bigRout$ is changed.
    
    Since each iteration of the for loop works regardless of $\tauout$, i.e. the initial value of $\bigRout$ going into the loop, the same logic applies to every $\cmid \in [C]$, accumulating the final summand for each. Thus the net result of the program is that for $\vout \in \bigvout$,
    \begin{equation*}
        \begin{split}
        \bigRout[\vout] &= \tauout[\vout] + \sum_{\cmid \in [C]} \sum_{\vmid \in V_{\cmid}} \sum_{\vin \in V_{\cin}} \tauin[\vin] \cdot \left(\sum_{W \in \Gamma_{\ell-1}(\vin,\vmid)}\mu(W)\right) \cdot \left(\sum_{W \in \Gamma_{\ell-1}(\vmid,\vout)}\mu(W)\right)\\
        &= \tauout[\vout] + \sum_{\vin \in V_{\cin}} \tauin[\vin] \cdot \sum_{\cmid \in [C]} \sum_{\vmid \in V_{\cmid}} 
        \left(\sum_{W \in \Gamma_{\ell-1}(\vin,\vmid)}\mu(W)\right) \cdot \left(\sum_{W \in \Gamma_{\ell-1}(\vmid,\vout)}\mu(W)\right) \\
        & =\tauout[\vout] + \sum_{\vin \in V_{\cin}} \tauin[\vin] \cdot \left(\sum_{W \in \Gamma_{\ell}(\vin,\vout)}\mu(W)\right) \\
        \end{split}
    \end{equation*}
    and so we add $\bigRin \cdot \calW_{\ell}$ to $\bigRout$ as required.
\end{proof}
\section{STCONN via Graph Decomposition}
\label{sec:graph_decomp_bbrs}
In this section we prove \cref{thm:main_stconn}, i.e. deciding connectivity in an $n$-vertex graph. To do this, we start by discussing important parameterizations of \cref{lem:main_subroutine} and what they say about the types of graphs and walks we can capture with them.

In order to run in polynomial time we need to choose $C$ and $\ell$ such that $(4C)^\ell \leq \poly(n)$, while for catalytic space $\frac{n}{2^{\Omega(\sqrt{\log n})}}$ we need to store $\bigRin,\bigRout,\bigRmid$ using $(m/C) \cdot \bigO(\log n)$ bits, where $m$ is the number of vertices of the graph in question. There are two extreme cases of interest:
\begin{itemize}
    \item when $C = 1$, we can tolerate any $2^\ell$-length walk for $\ell \leq \log m$, which in particular can cover every vertex in the graph in question; with respect to this graph, however, we only get small catalytic space if $m \leq \frac{n}{2^{\Omega(\sqrt{\log n})}}$, i.e. working on graphs with many fewer vertices than $n$.
    \item when $\ell \leq \sqrt{\log n}$ and $C \leq 2^{\bigO(\sqrt{\log n})}$, we can only handle short walks, but in exchange we use sufficiently little catalytic storage even when working on large graphs, such as $G$ or even larger graphs with $\bigO(n)$ vertices.
\end{itemize}
Thus we will have two phases of the algorithm as in~\cite{BBRS98}: 1) a \emph{long walks} algorithm which works on a vertex-sparsified version of $G$ whose edges represent walks of length $\lambda$ in $G$; and 2) a \emph{short walks} algorithm which runs on all of $G$ but only for walks of length $\lambda$. For this we will use both of the extreme settings of \cref{lem:main_subroutine} as discussed at the end of \cref{sec:push}; we discuss these in turn.

\subsection{Long Walks via Graph Sparsification}
In order to handle long walks in a graph, we need to reduce the number of vertices significantly while still not losing connectivity information about $G$. We do this through a randomized construction of a sufficiently large subset of vertices $U$, such that if any two vertices are connected in $G$, they are connected by taking short walks between the nodes in $U$. The following theorem will be proven in \cref{sec:hitting_sets} and may be of independent interest:

\begin{restatable}{thm}{uconstruct}
\label{thm:u_construct}
    There exists an algorithm $\calU$ such that for every directed graph $G = (V,E)$ on $n$ vertices, value $\lambda = \omega(\log(n))$, seed $\sigma \in \{0,1\}^{(6+\epsilon) \log (n)}$, and constant $\eps > 0$, it satisfies the following properties:
    \begin{itemize}
        \item For every $\sigma$, there exists a multiset $U_\sigma \subset V$ of size $\bigO \left( \frac{n \log(n)}{\lambda} \right)$, such that on input $i \in [|U_\sigma|]$ in binary, $\calU(1^n,\lambda,\sigma,\eps,i)$ outputs the $i^{\text{th}}$-vertex of $U_\sigma$ in time $\polylog{n}$ and space $\bigO(\log(n))$.
        \item With probability at least 0.99 over a uniformly random choice of $\sigma$, for every $u,v \in V$, if $u$ is connected to $v$, then there exists an $\ell$ and a sequence of vertices $\Tilde{\rho}_{u,v} = (u, s_1, s_2, \ldots, s_\ell, v)$, such that 1) $s_i \in U_\sigma$ for every $i \in [\ell]$, and 2) for each consecutive pair in $\Tilde{\rho}_{u,v}$, there exists a path in $G$ of length at most $\lambda$ connecting them.
    \end{itemize}
\end{restatable}

In order to utilize any multiset $U \coloneqq U_\sigma$ given by \cref{thm:u_construct}, we will need to have access to connectivity at distance at most $2^{\Theta(\sqrt{\log n})}$ in $G$. Note that \cref{lem:main_subroutine} counts walks of length \emph{exactly} $2^{\ell}$ for the given $\ell$, but this can be easily rectified by adding self-loops, at the cost of obtaining a value which is not the exact count of short walks in the original graph.

\begin{definition}\label{def:u_graph}
    Given graph $G = (V,E)$ on $n \in \mathbb{N}$ vertices, let $G'$ be the graph obtained by adding self-loops edges to all vertices, i.e. $V' = V$ and $E' = E \cup \{(v,v)\}_{v \in V'}$.
    
    We construct $H = (V_U,E_U)$ to be the graph whose vertex set $V_U$ is $U$ with multiplicity plus the vertices $s$ and $t$ (for convenience we assume $|V_U|$ is a power of two, and otherwise we add dummy vertices to make it such) and whose edge set $E_U$ is all $e = (\vin,\vout) \in V_U^2$ such that there exists a walk of length $\lambda$ from $\vin$ to $\vout$ in $G'$. We define the weight function $\mu(e)$ for all $e \in E_U$ to be the number of walks of length $\lambda$ from $\vin$ to $\vout$ in $G'$.\footnote{We alert the reader to the fact that for any $\mu(u,v)$ is \emph{not} the number of $u$--$v$ paths in $G$ of length at most $\lambda$; it is at least as large as this value and can be much greater. This will not cause an issue to our analysis, since all we need is a non-zero lower bound for connected vertices as well as a cap on the total size of any weight, the latter of which we address at the end of this section.}
\end{definition}

We can now compute our long walks algorithm, which will be our global algorithm modulo a few cleanup steps: 
\begin{lemma} \label{lem:long}
    Assume we have access to graph $G = (V,E)$ on $n \in \mathbb{N}$ vertices, prime $p = \poly(n)$, and seed $\sigma$. Let $\F := \mathbb{F}_p$, define $\lambda := 2^{\sqrt{(\epsilon/2.1)\log n}}$, and let $U \coloneqq U_\sigma$ be as given in \cref{thm:u_construct}. Define walk matrix $\calW$ with respect to graph $H = (V_U,E_U)$ and weight function $\mu$ as defined in \cref{def:u_graph}.

    Assume we have access to a catalytic subroutine $\shortpath(\bigvin, \bigvout; \bigRin, \bigRout)$, where $\bigvin, \bigvout \subseteq V_U$ and $\bigRin \in \F^{\bigvin}, \bigRout \in \F^{\bigvout}$ are any sections of the catalytic tape, whose effect is to add $\bigRin \cdot \calW_0$ to $\bigRout$.
    
    Then there exists a catalytic subroutine $\longpath(\bigvin, \bigvout; \bigRin, \bigRout)$, where $\bigvin, \bigvout \subseteq V_U$ and $\bigRin \in \F^{\bigvin}, \bigRout \in \F^{\bigvout}$ are any sections of the catalytic tape, whose effect is to add $\bigRin \cdot \calW_{\ell}$ to $\bigRout$, where $\ell := \log |V_U|$.

    If $\shortpath$ uses time $T$, free space $S$, and catalytic space $R$, then $\longpath$ uses time $T \cdot \bigO(n^2)$, free space $S + \bigO(\log n)$, and catalytic memory $|\bigvin| + |\bigvout| + 3|V_U| \cdot \bigO(\log p)$, namely the vectors $\bigRin, \bigRout$ plus three additional catalytic vectors $\bigRone, \bigRtwo, \bigRthree \in \F^{V_U}$.
\end{lemma}

\begin{proof}
    We have two tasks to accomplish. First, we will put aside our input and use \cref{lem:main_subroutine} to show how to compute the walk matrix for all of $U$; this will be fairly direct from the lemma using $\shortpath$ as our base case. Second, we will show how to translate from connectivity across all of $U$ to the actual sets $\bigvin$ and $\bigvout$, as well as switching our $\tau$ values from the internal memory of \cref{lem:main_subroutine} to the vectors $\bigRin$ and $\bigRout$; this will be handled by looping and cancellations.

    Consider \cref{lem:main_subroutine} with graph $H$, setting $C = 1$ and $\ell = \log \vert U \vert$. Putting aside $\bigRin$ and $\bigRout$ for the moment, we will use catalytic registers $\bigRone$, $\bigRtwo$, and $\bigRthree$ as its internal vectors $\bigRin$, $\bigRout$, and $\bigRmid$ respectively, each of length $m/C = |V_U|/1$.
    
    Our base case $P_0$ will be given by executing $\shortpath$ with $\bigvin = \bigvout = U$ and using the current input and output vectors as $\bigRin$ and $\bigRout$ respectively, plus auxiliary free space $S$ and catalytic space $R$. By the correctness of $\shortpath$, this exactly adds the weight $\vv{r_i} \cdot \calW_0$ to $\vv{r_j}$ for whichever $i,j \in \{in, out, mid\}$ it is invoked on; this is the same as adding $\mu(\vin,\vout)$ along each edge $(\vin,\vout) \in V_U^2$ as required.

    Thus \cref{lem:main_subroutine} gives us a program $P$ (and its inverse $P^{-1}$) which adds $\bigRin \cdot \calW_{\ell}$ to $\bigRout$ as required, doing so in time
    $$(T + \poly \log n) \cdot (4C)^{\log |V_U|} \leq T \cdot \left(\frac{n}{2^{\sqrt{\log n}}}\right)^2 + \left(\frac{n}{2^{\sqrt{\log n}}}\right)^2 \cdot \poly \log n \leq T \cdot n^2$$
    where the $\poly \log n$ terms comes from \cref{thm:u_construct} to run $\calU$ when calling $P_0$. Our space usage will be $S + \log |V_U| \cdot \log 4C$ from \cref{lem:main_subroutine} plus $\bigO(\log n)$ to compute the vertices of $U$ (from \Cref{thm:u_construct}), for a total of $S + \bigO(\log n)$. Lastly our catalytic space will be $R$ plus the vectors $\bigRone$, $\bigRtwo$, and $\bigRthree$, and by construction and the guarantee that \cref{lem:short} gives a catalytic subroutine, all catalytic memory is reset except $\bigRtwo$.

    In order to move to our original input and output registers, we will now create interfacing between $\bigRin$ and $\bigRone$ as well as between $\bigRtwo$ and $\bigRout$. Our final program is the following (in our in-line analysis we suppress indexing and let $\tau_a$ represent the initial value in its corresponding $r_a$, e.g. $\bigRin = \tauin$):
    \begin{enumerate}
        \item $\bigRout[\vout] \minuseq \bigRtwo[\vout] \quad \forall \vout \in \bigvout$ 
            \codecomment{$\bigRout = \tauout - \tautwo$}
        \item $P^{-1}$
            \codecomment{$\bigRtwo = \tautwo - \tauone \cdot \calW_{\ell}$}
        \item $\bigRout[\vout] \pluseq \bigRtwo[\vout] \quad \forall \vout \in \bigvout$ 
            \codecomment{$\bigRout = \tauout - \tautwo + (\tautwo - \tauone \cdot \calW_{\ell}) = \tauout - \tauone \cdot \calW_\ell$}
        \item $P$
            \codecomment{$\bigRtwo = \tautwo$}
        \item $\bigRone[\vin] \pluseq \bigRin[\vin] \quad \forall \vin \in \bigvin$
            \codecomment{$\bigRone = \tauone + \tauin$}
        \item $\bigRout[\vout] \minuseq \bigRtwo[\vout] \quad \forall \vout \in \bigvout$
            \codecomment{$\bigRout = \tauout - \tauone \cdot \calW_\ell - \tautwo$}
        \item $P$
            \codecomment{$\bigRtwo = \tautwo + (\tauone + \tauin) \cdot \calW_{\ell}$}
        \item $\bigRout[\vout] \pluseq \bigRtwo[\vout] \quad \forall \vout \in \bigvout$
            \codecomment{$\bigRout = \tauout - \tauone \cdot \calW_\ell - \tautwo + (\tautwo + (\tauone + \tauin) \cdot \calW_{\ell})$}\\
            \makebox{ } \codecomment{$= \tauout + \tauin \cdot \calW_{\ell}$}
        \item $P^{-1}$
            \codecomment{$\bigRtwo = \tautwo$}
        \item $\bigRone[\vin] \minuseq \bigRin[\vin] \quad \forall \vin \in \bigvin$
            \codecomment{$\bigRone = \tauone$}
    \end{enumerate}
    Thus our algorithm adds $\bigRin \cdot \calW_\ell$ to $\bigRout$ as required, running in time $4 \cdot (T \cdot n^2) + 6 \cdot \poly \log p \leq \bigO(T \cdot n^2)$.
\end{proof}

\subsection{Short Walks via Color Classes}
Our goal now is to compute $\shortpath$. By the definition of the weights of $H$, we have that computing $\calW_0$ with respect $H$ is equivalent to computing $\calW_{\log \lambda}$ with respect to $G'$ plus accounting for multiplicities; thus we utilize \cref{lem:main_subroutine} with this goal in mind.

\begin{lemma} \label{lem:short}
    Assume we have access to graph $G = (V,E)$ on $n \in \mathbb{N}$ vertices, prime $p = \poly(n)$, and seed $\sigma$. Fix $\lambda := 2^{\sqrt{(\epsilon/2.1) \log n}}$, define graphs $G'$ and $H$ as per \cref{def:u_graph}, and define $\F := \mathbb{F}_p$ and $C := \lambda$.
    
    Then there exists a catalytic subroutine $\shortpath(\bigvin, \bigvout; \bigRin, \bigRout)$, where $\bigvin, \bigvout \subseteq V_U$ are any sets of vertices of $H$ and $\bigRin \in \F^{\bigvin}, \bigRout \in \F^{\bigvout}$ are any sections of the catalytic tape, whose effect is to add $\bigRin \cdot \calW_0$ to $\bigRout$, where $\calW$ is defined with respect to $H$.
        
    $\shortpath$ runs in time $\bigO(|E| \cdot n^{\epsilon})$, uses free space $\bigO(\log n)$, and uses catalytic memory $|\bigvin| + |\bigvout| + \frac{3n}{2^{\sqrt{(\epsilon/2.1)\log n}}} \cdot \bigO(\log p)$, namely the vectors $\bigRin, \bigRout$ plus three additional catalytic vectors $\bigRone, \bigRtwo, \bigRthree \in \F^{n/C}$.
\end{lemma}

\begin{proof}
    Redefine $\calW$ to be with respect to $G'$ for the rest of this proof. We will do the same two-step procedure as before, first computing the walk matrix between each $(\cin,\cout) \in [C]^2$ pair, and second moving the relevant answers to the input and output vectors. The former will be very straightforward from \cref{lem:main_subroutine}, while the latter mostly follows the interfacing of \cref{lem:long}; however, we will need to additionally handle the changeover from $V_U$ (with multiplicity) to $V'$.
    
    Consider \cref{lem:main_subroutine} on graph $G'$ using weight function $\mu(e) = 1$, taking $C$ as given and $\ell := \log \lambda$. Again putting aside $\bigRin$ and $\bigRout$ for the moment, we will use catalytic registers $\bigRone$, $\bigRtwo$, and $\bigRthree$ as its internal vectors $\bigRin$, $\bigRout$, and $\bigRmid$ respectively, each of length $n/C$.
    
    We define a base case $P_0$, which will be checking connectivity in $G'$. Since $\ell = 0$, the set $\Gamma_0(\vin,\vout)$ contains the singleton edge $(\vin,\vout)$ of weight 1 if it is in $E'$ and is empty otherwise. Thus our program does the following:
    $$\bigRout[\vout] \pluseq \bigRin[\vin] \qquad \forall (\vin,\vout) \in E \cap (V_{\cin} \times V_{\cout}) $$
    which can be computed in time $2m\log C + |E'| \cdot \poly \log |\F|$ and space $2 \log(m/C) + \log C + \bigO(\log |\F|)$ by \Cref{prop:fieldwork}.
    
    Applying \cref{lem:main_subroutine}, we get that for any $\cin,\cout \in [C]$ we have an algorithm $P(\cin,\cout)$ (as well as its inverse) which successfully adds $\bigRone \cdot \calW_{\ell}$ to $\bigRtwo$, and which takes time
    $$\begin{array}{rcl}
        (4C)^{\ell} (2m \log C + |E'| \cdot \poly \log |\F|) & \leq & \bigO(|E'| \cdot (\cdot 2^{\sqrt{(\epsilon/2.1) \log m}})^{\sqrt{(\epsilon/2.1) \log m}})\\
        & &\quad \cdot 2^{2 \cdot \sqrt{(\epsilon/2.1) \log m}} \cdot \poly \log m\\
        & = & \bigO(|E| \cdot n^{\epsilon/1.01})
    \end{array}$$
    and uses free space
    $$\ell \log 4C + 3 \log(n/C) + \log C + \bigO(|\log \F|) \leq (\sqrt{(\epsilon/2.1) \log n})^2 + \bigO(\log n) = \bigO(\log n)$$
    as desired.

    Now we repeat the same interfacing procedure as in \cref{lem:long}. To bridge the notation gap between $V_U$ and $V'$, for any sets $\bigvin \subseteq V_U$ and $V_{\cin} \subseteq V'$ and a vertex $\vin \in \bigvin \cap V_{\cin}$, we let $\bigRin[\vin]$ mean the set of all entries $\bigRin[u_1] \ldots \bigRin[u_k]$ where $u_1 \ldots u_k \in V_U$ are all vertices whose base vertex from $V'$ is $\vin$, while $\bigRone[\vin]$ is the sole entry indexed by $\vin$; we do the same for $\bigvout \cap V_{\cout}$.
    
    Our algorithm $\shortpath$ will directly invoke $P(\cin,\cout)$ over each pair of color classes and ignore all answers except those from vertices $\bigvin \cap V_{\cin}$ to vertices $\bigvout \cap V_{\cout}$, which will precisely cover each pair of vertices $(\vin,\vout) \in \bigvin \times \bigvout$ once.
    \begin{itemize}
        \item for all $\cin, \cout \in [C]$:
        \begin{enumerate}
            \item $\bigRout[\vout] \minuseq \bigRtwo[\vout] \quad \forall \vout \in \bigvout \cap V_{\cout}$
            \item $P^{-1}(\cin,\cout)$
            \item $\bigRout[\vout] \pluseq \bigRtwo[\vout] \quad \forall \vout \in \bigvout \cap V_{\cout}$
            \item $P(\cin,\cout)$
            \item $\bigRone[\vin] \pluseq \bigRin[\vin] \quad \forall \vin \in \bigvin \cap V_{\cin}$
            \item $\bigRout[\vout] \minuseq \bigRtwo[\vout] \quad \forall \vout \in \bigvout \cap V_{\cout}$
            \item $P(\cin,\cout)$
            \item $\bigRout[\vout] \pluseq \bigRtwo[\vout] \quad \forall \vout \in \bigvout \cap V_{\cout}$
            \item $P^{-1}(\cin,\cout)$
            \item $\bigRone[\vin] \minuseq \bigRin[\vin] \quad \forall \vin \in \bigvin \cap V_{\cin}$
        \end{enumerate}
    \end{itemize}
    The above program has the same analysis as in \cref{lem:long}, and so we only need confirm what happens to vertices with multiplicity. Consider a set of vertices $u_1 \ldots u_k \in V_U$ which are copies of the same vertex $\vin \in V'$, and consider some $\vout$. Restricting our view to $(\vin,\vout)$ and letting $\vv{\tau}$ represent the initial values in their respective $\vv{r}$ as before, since we add all $\tauin[u_i]$ to $\tauone[\vin]$, the value added to $\bigRout[\vout]$ is
    $$((\tauin[u_1] + \ldots + \tauin[u_k]) \cdot \calW_{\ell})[\vout] = (\tauin[u_1] \cdot \calW_{\ell})[\vout] + \ldots + (\tauin[u_k] \cdot \calW_{\ell})[\vout]$$
    which is exactly as required for computing $\calW_0$ over graph $H$.
    Thus looping over all $(\cin,\cout) \in [C]^2$, our algorithm successfully adds $\bigRin \cdot \calW_\ell$ to $\bigRout$ and resets all other catalytic memory, running in time
    $$C^2 \cdot [4 \cdot \bigO(|E| \cdot n^{\epsilon/1.01}) + 6 \cdot |V_U| \cdot \bigO(\log |U|) \cdot \bigO(\log p)] \leq \bigO(|E| \cdot n^{\epsilon})$$
    as desired.
\end{proof}

\subsection{Final Algorithm}
Our $\longpath$ algorithm, using $\shortpath$ as a subroutine, gives us everything we need to compute $\stconn$, with the remaining issues being to choose $U$ and $p$ so as to ensure correctness and efficiency.
\begin{proof}[Proof of \cref{thm:main_stconn}]
    For the moment assume we have access to prime $p = \poly(n)$ and seed $\sigma$. Define $\F = \F_p$, define $\lambda := 2^{\sqrt{(\epsilon/2.1)\log n}}$, and let $U \coloneqq U_\sigma$, and subsequently $H = (V_U,E_U)$ be as given in \cref{thm:u_construct} and \cref{def:u_graph} respectively.

    Define $\bigvin := \{s\}$ and $\bigvout := \{t\}$, and we will have $\bigRin, \bigRout \in \F$ come from the free work tape instead of the catalytic tape, being initialized to $1$ and $0$ respectively. Thus running $\longpath$ gives a final value of
    $$0 + 1 \cdot \sum_{W \in \Gamma_{\log |V_U|}(s,t)} \mu(W)$$
    where $\Gamma$ is taken over $H$.\footnote{Note that because each vertex in $G'$ has self-loops and hence walks of length $\lambda$ to itself, the vertices of $H$ also each have self-loops of the appropriate weight.} Moving to $G'$, every walk in $G'$ has weight 1 since all edges have weight 1, and so our final value will be
    $$\sum_{W \in \Gamma_{\log |V_U|}(s,t)} \mu(W) = \sum_{W \in \Gamma_n^*(s,t)} 1 = |\Gamma_n^*(s,t)|$$
    where $\Gamma_n^*(s,t)$ is the set of walks from $s$ to $t$ in $G'$ of length $n$ which pass through nodes in $U$ at most every $\lambda$ steps.
    
    Assuming we have our modulus $p$ and our seed $\sigma$ for \cref{thm:u_construct} on the free work tape, putting \cref{lem:long} and \cref{lem:short} together our total runtime is
    $$\bigO(|E| \cdot n^{\epsilon}) \cdot \bigO(n^2) = \bigO(|E| \cdot n^{2+\epsilon})$$
    our free space usage is
    $$\bigO(\log p + \log |U| + \log n) = \bigO(\log n)$$
    and our catalytic space usage is
    $$6 \cdot \left(\frac{n}{2^{\sqrt{(\epsilon/2.1) \log n}}}\right) \cdot \bigO(\log p) = \frac{n}{2^{\Omega(\sqrt{\log n})}}$$
    Because \cref{lem:long} does not change any values of the catalytic tape besides the output registers, which we have put on the free work tape, our catalytic memory is successfully returned to its original configuration.

    Our deterministic algorithm simply loops over all descriptions of $p$ and $\sigma$, and accepts iff any choice gives a non-zero answer, while our randomized algorithm chooses $p$ and $\sigma$ uniformly and thus incurs no time overhead. Clearly if there is no $s$--$t$ path, no choice of $p$ and $\sigma$ returns a non-zero answer, and so assume otherwise. By \cref{thm:u_construct}, $U$ is such that if there is any $s$--$t$ path in $G$, then $\Gamma^*_n(s,t)$ will be non-empty with probability 0.99. Conversely, since every walk in $\Gamma^*_n(s,t)$ is an $s$--$t$ walk in $G'$ of length $n$, we have
    $$|\Gamma^*_n(s,t)| \leq n^n = 2^{n \log n}$$
    since at every step there are at most $|G'| = n$ choices of the next vertex for the walk. Thus by \cref{prop:rand_prime}, $|\Gamma_n^*(s,t)|$ will be non-zero modulo $p$ with probability 0.99 over a random choice of modulus $p \in [n \log^2 n, 3n \log^2 n]$ (which can be picked at random using a standard procedure of trial-and-error in time $\polylog{n}$). In total, over a random choice of $p$ and $U$ we have a 0.98 chance of accepting whenever $s$ and $t$ are connected.
\end{proof}

\section{Constructing Representative Multisets for Directed Graphs}
\label{sec:hitting_sets}
In this section we prove \cref{thm:u_construct}. We start by recalling the statement.
\uconstruct*
\noindent We call multisets $U_\sigma$ that satisfy Item (2) as a $\lambda/2$-\textit{representative multiset}.

Our construction of representative multisets will rely on the classical and well-used notion of \emph{expander graphs}. An undirected (multigraph) $D$ is an $(N,\Delta,\alpha)$-\textsf{expander graph} if it is a $\Delta$-regular connected graph on $N$ vertices, and the second-highest eigenvalue of its normalized adjacency matrix is bounded by $\alpha \in [0,1]$ (in absolute value). Moreover, the family $\{D_N\}$ is said to have a (fully) \textsf{explicit} construction, if there exists a deterministic algorithm that takes a vertex $v \in \Bits^{\log(N)}$, and a number $i \in [\Delta]$ as inputs (in binary), and computes the $i^{\text{th}}$ neighbor of $v$ in time $\polylog{N}$ and space $\bigO(\log N)$. We use the following explicit expander construction.
\begin{lemma}[\cite{LPS88,Morg94}]
\label{lem:lps_expander}
    For every prime $p$ and every positive integer $k$, there exists an infinite family of explicit graphs $\{D_N\}$, such that each $D_N$ is an $(N, \Delta, \alpha)$-expander for $\Delta = p^k + 1$ and $\alpha \leq \frac{2\sqrt{\Delta-1}}{\Delta}$. 
\end{lemma}

\noindent This is followed by the definition of an expander random walk.
\begin{definition}[Expander random walk]
    \label{def:exp_walk}
    For any $(N,\Delta,\alpha)$-expander $D = (V_D,E_D)$, an \textbf{expander random walk} $(X_1, \dots, X_K)$ of length $K-1$, is one where $X_1$ is the uniform distribution over $[N]$, and for every $i \in [K-1]$, $X_i$ is the uniform distribution over the $\Delta$-sized neighborhood of the $(i-1)^{\text{th}}$-vertex in the random walk.
\end{definition} 

For any $\mathcal{B} \subset V_D$, its density is defined as $\frac{\vert \mathcal{B} \vert}{V_D}$. Then, an expander random walk of length $K$ over $[N]$ behaves similar to uniformly sampling $K$ elements from $[N]$, in the following sense.
\begin{lemma}[\cite{AKS87_derand} (see also \cite{AS16_book})]
    \label{lem:exp_rw}
    Let $D$ be an $(N,\Delta,\alpha)$-expander, and let $\mathcal{B} \subset V_D$ of density $\beta$, for some $\beta > 0$. Then, for every $K$, the probability that an expander random walk $(X_1, \dots, X_K)$ of length $(K-1)$ in $D$ never leaves $\mathcal{B}$, is at most $\beta (\beta + \alpha \cdot (1-\beta))^{K-1}$.
\end{lemma}

Next, we have the following lemma about a pairwise independent hash function family $\calH_{m,n}$. We say that a hash function $h$ \textsf{hits} a set $B \subseteq [n]$, if there exists an $i \in [m]$, such that $h(i) \in B$, i.e., if $\mathsf{Im}(h)$ intersects $B$. We prove that a randomly chosen $h$ from $\calH_{m,n}$, where $m$ is $\bigO(n/\lambda)$, hits any $\lambda$-sized subset of $[n]$ with large probability. 
\begin{lemma}
    \label{lem:hit_small_subsets}
    For any constant $0 < \beta < 1/2$, let $m = \frac{n(1/\beta)}{\lambda}$, and let $\calH_{m,n}$ be a pairwise independent hash function family from $[m]$ to $[n]$. Then, for any $B \subset [n]$ of size $\lambda$ we have,
    \begin{equation*}
        \Pr_{h \sim \calH_{m,n}} [h \text{ hits } B] \geq 1 - \beta.
    \end{equation*}
\end{lemma}

\begin{proof}
    Let $h$ be a random hash function sampled from $\calH_{m,n}$. For every $i \in [m]$, define the indicator random variable $X_i$ that is $1$ iff $h(i) \in B$. Note that the the random variables $X_1, \dots, X_m$ are pairwise independent from the definition of $\calH_{m,n}$. 
    
    Let $X = \sum_{i=1}^m X_i$, be the random variable that counts the size of $\mathsf{Im}(h) \cap B$. Observe that $$\text{E}[X] = \sum_{i=1}^m \text{E}[X_i] = \sum_{i=1}^m \Pr[h(i) \in B] = \frac{m\lambda}{n},$$ and from the pairwise independence of the $X_1, \dots, X_m$, $$\text{Var}(X) = \sum_{i=1}^m \text{Var}(X_i) = \frac{m\lambda}{n} \left( 1 - \frac{\lambda}{n} \right) = \text{E}[X]\left( 1 - \frac{\lambda}{n} \right).$$
    Upper bounding the probability that $h$ does not hit $B$, i.e., $X = 0$, by $\beta$, we prove the lemma.
     \begin{equation*}
        \begin{split}
            \Pr_{h \sim \calH_{m,n}} [X = 0] &\leq \Pr_{h \sim \calH_{m,n}} \left[ \vert X - \text{E}[X] \vert \geq \text{E}[X] \right] \leq \frac{\text{Var}(X)}{(\text{E}[X])^2} \\
            &= \frac{\text{E}[X] \left( 1 - \lambda/n \right)}{(\text{E}[X])^2} \leq \frac{1}{\text{E}[X]} \\
            &\leq \frac{n}{m\lambda} = \beta. \\
        \end{split}
    \end{equation*}
    The second inequality follows from Chebyshev's inequality\footnote{For any random variable $X$ with finite non-zero variance, and for any $a > 0$, $\Pr_X \left[ \vert X - \text{E}[X] \vert \geq a \right] \leq \frac{Var(X)}{a^2}$.} \cite{vadhan_psuedo_2012}, and the final probability calculation comes from our setting of $m = \frac{n(1/\beta)}{\lambda}$.
\end{proof}

We now show a probabilistic construction of representative multisets. Let $D$ be an explicit $(N,D,\alpha)$-expander, where $N = \{0,1\}^{2\log(n)}$, and each vertex describes a function in $\calH_{m,n}$. Picking up the hash functions obtained from the vertices of a random walk of length $K = \bigO(\log(n))$ on $D$, we prove that for every pair of vertices $u,v$ connected by a path of length $\lambda/2$, at least of one of the hash functions hits some walk visiting $\lambda/2$ distinct vertices between $u$ and $v$, with high probability. Finally, this property implies that the union of their images is a $\lambda/2$-representative multiset, with high probability.\footnote{In fact, notice that the two properties are equivalent; for $u$ and $v$ connected by a path of length $\lambda/2$, item(2) implies item (1) trivially.}

\begin{proof}[Proof of \cref{thm:u_construct}]
Let $0 < \beta < 1/2$ be a small constant that we set later in the proof. For a fixed $\lambda = \omega(\log(n))$ and $\eps > 0$, let $m = \frac{2n(1/\beta)}{\lambda}$, and $\calH_{m,n}$ be a pairwise independent hash function family from $[m]$ to $[n]$, given by \cref{lem:pwi_hash_props}.

Let  $D = (V_D,E_D)$ be an $(N,\Delta,\alpha)$-expander, where $N = \vert \calH_{m,n} \vert$, and $\Delta$ is chosen such that $\alpha \leq 2\beta$. Specifically, we use the explicit construction from \cref{lem:lps_expander}, by setting $\Delta = \frac{1}{\beta^2} + 1$ (assume that $1/\beta$ is a prime power), which in turn sets $\alpha \leq \frac{2\sqrt{\Delta-1}}{\Delta} \leq 2\beta$.
 
For any seed $\sigma$, the multiset $U_\sigma$ is constructed as follows: for $K = \left( 2+\frac{\eps}{8} \right)\log(n)/\log \left( \frac{1}{3\beta} \right)$, use $\sigma$ to specify the vertices $(X_1, \dots, X_K)$ as a walk of length $K-1$ on the expander $D$ (as described in \cref{def:exp_walk}). Interpreting each $X_i$ as the description of a hash function $h_i$ from $\calH_{m,n}$, the set $U_\sigma$ is defined as $\bigcup_{i=1}^K \mathsf{Im}(h_i)$, i.e., the union of the images of all $h_i$'s. Clearly, the size of $U_\sigma$ is $\bigO \left( \frac{n\log(n)}{\lambda}\right)$.

\vspace{0.1in}
We next show that a seed length of $(6+\eps)\log(n)$ for $\sigma$, suffices to construct $U_\sigma$. Firstly, observe from \cref{lem:pwi_hash_props} that the description length of any $h \in \calH_{m,n}$ is $\log(N) = \max\{\log(m),\log(n)\}+\log(n) = 2\log(n)$ bits, since $m < n$ for a constant $\beta$. The vertices $h_1, \dots, h_K$ on a walk in $D$ can be specified using $\log(N) + K \cdot \log(\Delta)$, which is equal to $2\log(n) + \frac{\log(n)(2+\eps/8)}{\log \left( \frac{1}{3\beta} \right)}\cdot\log \left( \frac{1}{\beta^2}+1\right)$. Since $\lim_{\beta \rightarrow 0^+} \frac{\log\left( \frac{1}{\beta^2}+1\right)}{\log \left( \frac{1}{3\beta} \right)} = 2$ we can set $\beta>0$ to be a small enough constant so that $\frac{\log\left( \frac{1}{\beta^2}+1\right)}{\log \left( \frac{1}{3\beta} \right)} < (2+\eps/4)$. Thus, the walk on $D$ can be described using at most $2\log(n) + ((2+\eps/8)(2+\eps/4))\log(n)$ bits, which is at most $(6+\eps) \cdot \log(n)$, and a seed length of $(6+\eps)\log(n)$ is sufficient.

\vspace{0.1in}
Next, we prove that with high probability over the choice $\sigma$ (or a random walk of length $K$ on $D$), for \textit{every} pair $u,v \in V$ that is connected by a path from $u$ to $v$ of length $\lambda/2$, there exists at least one walk $\rho_{u,v}$ that visits $(\lambda/2)-1$ distinct vertices, and is hit by $U_\sigma$ (i.e., one of the hash functions that form $U_\sigma$).
\begin{lemma}
\label{lem:expander_rw_good}
For a small enough constant $0 < \beta < 1/2$, let $D = (V_D,E_D)$ be an $(N,\Delta,\alpha)$-expander, where $N = \vert \calH_{m,n} \vert$, and $\Delta = \frac{1}{\beta^2} + 1$. 

Then, for any $G = (V,E)$, any $\eps > 0$, for $K = \frac{\left( 2+\frac{\eps}{8} \right)\log(n)}{\log \left( \frac{1}{3\beta} \right)}$, we have
$$\Pr_{\sigma \sim \{0,1\}^{(6+\eps)\log(n)}} \left[ U_\sigma \text{ is a } \lambda/2 \text{-representative set for } G \right] \geq 1 - \frac{1}{n^{\eps/8}}$$
\end{lemma}

\begin{proof}
    Let $\lambda' = \lambda/2$. Pick any $u,v \in V$, for which there exists a $\lambda'$-length path from $u$ to $v$, and fix some walk $\rho_{u,v}$ that visits distinct vertices $p_1, \dots, p_{\lambda'-1}$ from $u$ to $v$. Define the set $P_{u,v} = \{u,p_1, \dots, p_{\lambda'-1}\}$.

    Let $\mathcal{B}$ be the subset of $V_D$ in the expander graph $D$ that contains hash functions in $\calH_{m,n}$ that do not hit $P_{u,v}$. Applying \cref{lem:hit_small_subsets} for subsets of size $\lambda'$ and $m = \frac{n(1/\beta)}{\lambda'}$, we see that $\mathcal{B}$ has density at most $\beta$. In addition, observe that from the expander construction in \cref{lem:lps_expander}, $\alpha \leq \frac{2\sqrt{\Delta-1}}{\Delta} \leq \frac{2}{\sqrt{\Delta}} \leq 2\beta$.

    From the property of an expander random walk given by \cref{lem:exp_rw}, we have
    \begin{equation*}
    \begin{split}
        \Pr_{(h_1, \dots, h_K) \sim (X_1, \dots, X_K)} &\left[ \forall i \in [K], \forall x \in [m], h_i(x) \notin P_{u,v} \right]  \\ 
        &\leq\beta \cdot (\beta + \alpha \cdot (1 - \beta ))^{K-1} \\
        &\leq \beta \cdot (\beta + 2\beta \cdot (1 - \beta ))^{K-1} \leq (3\beta)^{K}\\
    \end{split}
    \end{equation*}
    In the analysis above, we use $\sigma \in \{0,1\}^{(6+\eps)\log(n)}$ interchangeably with the $K$ hash functions it visits in the random walk over $D$.
    
    Choosing $K = \frac{\left( 2+\frac{\eps}{8} \right)\log(n)}{\log \left( \frac{1}{3\beta} \right)}$, the probability that $P_{u,v}$ intersects with $U_\sigma$ is at most $\frac{1}{n^{2+\eps/8}}$. Next, we apply a union bound over all pairs of vertices $u,v \in V$, for which $u$ is connected to $v$ by a $\lambda'$-length path. Since there at most $n^2$ such pairs, with probability at least $1 - \frac{1}{n^{\eps/8}}$ over the choice of $\sigma$, for each such $u,v$, there exists a walk $\rho_{u,v}$ that visits $\lambda'-1$ distinct vertices, whose vertex set is hit by $U_\sigma$, and thus, with this probability, $U_\sigma$ is a $\lambda/2$-representative set for $G$.
\end{proof}

To prove Item (1), we store the seed $\sigma$ on the work tape of size $\bigO(\log(n))$. On input $i \in [\vert U_\sigma \vert]$ in binary, $\calU$ first computes $a = \lfloor i/m \rfloor$ and $b = i\pmod m$, followed by which it outputs $h_a(b)$. From \cref{prop:fieldwork}, both operations can be performed in $\polylog{n}$ time and $\bigO(\log(n))$ space. For the latter computation, $\calU$ obtains $h_1$ from the first $2\log(n)$ bits in $\sigma$, and starts a walk on $D$, through the remaining bits on $\sigma$. Using the explicitness of $D$ from \cref{lem:lps_expander} to obtain the $a^{\text{th}}$ hash function from the random walk, $h_a(b)$ can now be computed using the algorithm from \cref{lem:pwi_hash_props}. From these lemmas, each of these operations can be done in $\polylog{n}$ time and $\bigO(\log(n))$ space, and thus Item (1) follows.

\vspace{0.1in}
Finally, we prove Item (2) of the theorem. Pick any $u,v \in V$ that is connected by a path $\rho_{u,v} = (u, p_1, \dots, p_{\ell'},v)$. Now, consider a partition of $\rho_{u,v}$ into subsequences, where $\rho^1$ is the subpath of $\rho_{u,v}$ of length $\lambda/2$ from $u$ to $q_{\lambda/2}$, and each subsequent $\rho^i$ is the subpath of $\rho_{u,v}$ from $q_{(i-1)\lambda/2}$ to $q_{i\lambda/2}$ visiting the next $\lambda/2$ vertices that appear after $q_{(i-1)\lambda/2}$. Finally, $\rho^{\ell+1}$ contains at most $\lambda/2$ vertices that appear from $q_{\ell \lambda/2}$ to $v$ in $\rho_{u,v}$.

With high probability over $\sigma$, $U_\sigma$ is a $\lambda/2$-representative set, and so, for each $i$, there exists a walk $\Tilde{\rho}^i$ from $q_{(i-1)\lambda/2}$ to $q_{i\lambda/2}$ that visits $(\lambda/2)-1$ distinct vertices (which could even be $\rho^i$), whose vertex set $P^i$ (that includes $u$) is hit by $U_\sigma$. For each $i$, pick an $s_i \in U_\sigma$ that intersects $P^i$. 

Now, consider the sequence $(u, s_1, \dots, s_{\ell}, v)$. For each pair $(s_i,s_{i+1})$, there exists a walk from $s_i$ to $q_{i\lambda/2}$ in $\Tilde{\rho}^i$ that visits at most $(\lambda/2)-1$ distinct vertices, which in turn implies a path of length at most $\lambda/2$ between them. Similarly, we obtain a path from $q_{i\lambda}$ to $s_{i+1}$ in $\Tilde{\rho}^{i+1}$ of length at most $\lambda/2$. Together, there exists a path of length at most $\lambda$ from $s_i$ to $s_{i+1}$. In addition, observe that the path from $u$ to $s_1$ has length at most $\lambda/2$, and the one from $s_{\ell\lambda/2}$ to $v$ has length at most $\lambda$, as the final set $P^{\ell+1}$ might not be large enough to be hit by $U_\sigma$.
\end{proof}

\section{Edit Distance, Longest Common Subsequence and Discrete Fr\'{e}chet Distance}

In this section, we focus on computing edit distance, longest common subsequence and discrete Fr\'{e}chet distance. We prove the next two theorems.

\begin{theorem}\label{thm-edit}
    The length of the longest common subsequence and edit distance of two strings of length $n$ can be computed in time $\poly(n)$, using $\bigO(\log n)$ workspace and $\frac{n}{2^{\Omega(\sqrt{\log n})}}$ catalytic space.
\end{theorem}

\begin{theorem}\label{thm-frechet}
    Discrete Fr\'{e}chet distance of two point sequences of length $n$ with distance oracle whose values are in the range $\{0,\ldots,\poly(n)\}$ can be computed in time $\bigO(n^{4+\varepsilon})$, using $\bigO(\log n)$ workspace and $\frac{n}{2^{\Omega(\sqrt{\log n})}}$ catalytic space.
\end{theorem}

All of these results follow from the same procedure for computing (weighted) reachability on a directed grid graph. This procedure is presented in the next section.

\subsection{Grid graphs reachability}

For the remainder of this section, we consider only fields of the form $\F_p = \Z/p\Z$ for some prime $p$.

For the purpose of our algorithm, we extend the grid-like graphs used for computing the metrics to be layered, denoted by $\layergrid_{n, n}$. 
Formally, for $n\in\N$ we define $\layergrid_{n, n}$ as a directed graph such that
\begin{align*}
    V(\layergrid_{n, n})=&\,\{(i,j)\colon 0\leq i\leq n, 0\leq j\leq n\},\\
    E(\layergrid_{n, n})=&\,\{((i,j), (i+1,j))\colon 0\leq i\leq n-1, 0\leq j\leq n\}\,\cup\\
    &\,\{((i,j), (i+1,j+1))\colon 0\leq i\leq n-1, 0\leq j\leq n-1\}\,\cup\\
    &\,\{((i,j), (i+1,j-1))\colon 0\leq i\leq n-1, 1\leq j\leq n\}.
\end{align*}

Given a directed acyclic layered graph $G=(V,E)$ with weighted edges $\mu: V\times V\rightarrow\N$, we recall that the weight of an $s,t$-walk $W=(s=v_0,v_1,\ldots,v_k=t)$ is $\mu(W) = \prod_{i=0}^{k-1} \mu(v_i,v_{i+1})$.
In fact, as we consider only layered graphs, and in particular $\layergrid_{n,n}$, all walks from any $u$ to any $v$ are paths.
The restriction to paths will not cause any problems in the following algorithms, as in layered graphs, all paths between two vertices have the same length and recursively splitting the problem with respect to the layers decomposes the problem into disjoint subproblems.

We thus define the total path weight\footnote{One can also view this as $\tauin\calW_\ell$ for a particular $\ell$ in Definition~\ref{def:path_flow} with $\tauin$ set to be the unit vector corresponding to $s$.} from $s$ to $t$ as $\calW_{s,t}=\sum_{P\, s,t\text{-path}}\mu(P)$.
In line with the flow-based algorithm of~\cite{cook_pyne_26_efficient}, our goal is to catalytically compute the total path weight from $(0,0)$ to $(n,n)$ in a weighted $\layergrid$.
We make use of two catalytic subroutines, which we call $\inneralg$ and $\outeralg$.
To save space, we split each layer of the graph into a few color classes (which will be intervals in each of the layers), and then compute the total path weights recursively on the layers within the color classes up to the base case, where we will use only a single color class on which we will recurse slightly differently.
The former recursive computation is the $\outeralg$ procedure, the latter is the $\inneralg$ procedure. We first start with describing the $\inneralg$ procedure.
\begin{lemma}\label[lemma]{lem:layered_inner}
    There exists a catalytic subroutine $\inneralg^{\ell}(\bigRin, \bigRout; \lin, \lout, \lmid; \itop, h)$, where $\bigRin, \bigRout \in \F^{h}$ are any sections of the catalytic space and $\lin, \lout, \lmid$ are three $\log(n)$-bit free space variables, such that for $\lout-\lin \leq 2^\ell$, given a weighted graph $\layergrid_{n,n}$, where the weights are given by an oracle $\weightoracle$ that on input $e$ outputs the weight of the edge $e$ modulo $p$, if the initial contents of $\bigRin$ and $\bigRout$ were $\tauin$ and $\tauout$ respectively, then $\bigRin$ remains unchanged, and $\bigRout$ becomes $\tauout + \tauin \cdot \calW$, where the corresponding parts of the layers $\lin, \lout$ are the vertices with indices $\{\itop, \ldots,\itop+h-1\}$.

    This procedure uses $4^{\ell-1}\cdot \bigO(h)$ weight oracle queries, $4^{\ell}\cdot \bigO(h\cdot \polylog{|\F|}) + 5\cdot \sum_{i=1}^{\ell}4^{i}\leq 4^\ell \cdot (\bigO(h\cdot \polylog{|\F|}) + 20)$ time, $(\ell-1)\cdot h\cdot \log(|\F|)$ catalytic space (not including the input registers), and $(\ell-1)\cdot 4 + \bigO(\log |\F| + \log n)$ free space.
\end{lemma}
\begin{proof}
    For simplicity of exposition, we assume that $n, h, \lout-\lin$ are all powers of two. (This can be removed at the cost of a few bits of free space per level of recursion.)
    
    We define the procedure recursively.
    For $\ell=0$, we want to compute the flow path for two neighboring layers.
    We can do this computation as follows: for each vertex $(\lin,i)$ in the input layer and each of its at most three neighbors $(\lout,j)$ in the output layer, we add $\bigRout[j] \pluseq \bigRin[i]\cdot \mu((\lin,i), (\lout,j))$ using a single call to the oracle $\weightoracle$, free space $\bigO(\log |\F|)$ (which can be reused for each pair) and time $\bigO(\polylog{|\F|})$ per pair (the time comes from accessing the entry of the vector according to our methodology in the appendix). The catalytic space used is already provided as input, so we do not use any more.
    Thus, as there are $\bigO(h)$ edges we must iterate over, we use time $\bigO(h\cdot \polylog{|\F|})$, $\bigO(h)$ calls to the oracle and free space $\bigO(\log |\F| + \log n)$, where the second term is the counter required for us to iterate over all the edges.
    All of catalytic space used is provided as part of the input.

    We use $\bigvin, \bigvmid, \bigvout$ to refer to the sets of vertices corresponding to the respective layered color classes.
    We also note that as $\inneralg^{\ell-1}$ is a catalytic subroutine, we can also use its inverse\footnote{As in previous statements, we may run the same program backwards while reverting the instructions dealing with catalytic space. We also need to take care of the parameters $\lin, \lout, \lmid$, but we can use identical ideas from the original algorithm to do so.}.
    We are now ready to define $\inneralg^{\ell}(\bigRin, \bigRout; \lin, \lout, \lmid; \itop, h)$ as follows:
    Taking the first so far unused allocated catalytic register vector $\bigRmid \in \F^{h}$, we perform the following:

    \begin{enumerate}
        \item $\lmid = \frac{\lin+\lout}{2}$
        \item $\inneralg^{\ell-1}(\bigRmid, \bigRout; \lmid, \lout, \lin; \itop, h)^{-1}$\codecomment{$\bigRout[\vout]\minuseq \sum\limits_{\vmid\in \bigvmid}\bigRmid[\vmid]\cdot \calW_{\vmid,\vout}\quad\forall \vout$}
        \item $\lin = 2\cdot \lmid-\lout$
        \item $\inneralg^{\ell-1}(\bigRin, \bigRmid; \lin, \lmid, \lout; \itop, h)$\codecomment{$\bigRmid[\vmid]\pluseq \sum\limits_{\vin\in \bigvin}\bigRin[\vin]\cdot \calW_{\vin,\vmid}\quad\forall \vmid$}
        \item $\lout = 2\cdot \lmid-\lin$
        \item $\inneralg^{\ell-1}(\bigRmid, \bigRout; \lmid, \lout, \lin; \itop, h)$\codecomment{$\bigRout[\vout]\pluseq \sum\limits_{\vmid\in \bigvmid}\bigRmid[\vmid]\cdot \calW_{\vmid,\vout}\quad\forall \vout$}
        \item $\lin = 2\cdot \lmid-\lout$
        \item $\inneralg^{\ell-1}(\bigRin, \bigRmid; \lin, \lmid, \lout; \itop, h)^{-1}$\codecomment{$\bigRmid[\vmid]\minuseq \sum\limits_{\vin\in \bigvin}\bigRin[\vin]\cdot \calW_{\vin,\vmid}\quad\forall \vmid$}
        \item $\lout = 2\cdot \lmid-\lin$
    \end{enumerate}

    By induction, we require $4\cdot (4^{\ell-1}\cdot \bigO(h\cdot \polylog{|\F|}) + 5\cdot \sum_{i=1}^{\ell-1}4^{i}) + 5 = 4^{\ell}\cdot \bigO(h\cdot \polylog{|\F|}) + 5\cdot \sum_{i=1}^{\ell}4^{i}$ time.
    
    We only require 4 bits of new free space within the recursion for the instruction counter. Otherwise, we only use the previously stored $\ell$s, therefore the free space required remains $\bigO(\log n)$ for the input parameters and by induction, $4 + (\ell-2)\cdot 4 + \bigO(\log |\F| + \log n)=(\ell-1)\cdot 4 + \bigO(\log |\F| + \log n)$.
    We use a new part of the catalytic space, so using induction, we use $h\cdot \log(|\F|) + (\ell-2)\cdot h\cdot \log(|\F|)$, and thus the total catalytic space used is $(\ell-1)\cdot h\cdot \log(|\F|)$ plus the input registers.
    Finally, the number of weight oracle queries is again by induction $4\cdot 4^{\ell-2}\cdot \bigO(h)=4^{\ell-1}\cdot \bigO(h)$.

    The correctness of the computation adding the weighted path-flow follows almost identically to the proof of~\Cref{lem:main_subroutine}.
    We therefore only focus on ensuring that the layer registers $\lin, \lout, \lmid$ have correct values so that we recurse correctly while using only $\bigO(\log n)$ space.
    Our invariant will be that after returning from a recursive call $\inneralg^{\ell}(\bigRin, \bigRout; \lin, \lout, \lmid; \itop, h)$, the two variables $\lin$ and $\lout$ will have their values identical to the values at the start of the procedure.
    In the case when $\ell=0$, we never change these values, so the invariant holds.
    
    Next, we consider $\ell\geq 1$.
    The first operation only computes the index of the middle layer.
    After every even operation, we use $\lmid$ and one of $\lin, \lout$, as the in/out layers for the recursive call, while the other $\ell_{b}$ is used as the mid-layer.
    Therefore, the induction hypothesis guarantees that $\lmid$ and one of the other two will be returned with the values set as they were when the subroutine was called.
    In the following instructions, we then reconstruct the value that may have been changed so that all three variables have the same values as at the start of the procedure.
\end{proof}

We also note that we do not optimize the number of catalytic vectors.
Indeed, one can only use three vectors and swap them around similarly to $\lin, \lout, \lmid$.

\begin{definition}[Layered color classes]
    Let $C$ be a parameter describing the number of color classes.
    Consider $V(\layergrid_{n,n})=\{(i,j)\colon 0\leq i\leq n, 0\leq j\leq n\}$.
    We define the \emph{layer} of each vertex as its first coordinate.
    We further define $\chi: V \rightarrow [C]$ and $\delta: V \rightarrow [n/C]$ as $\chi((i,j)) = \lceil (j+1)/(n/C)\rceil$ and $\delta((i,j)) = j \bmod (n/C)$ as the color class and index within the color class respectively.

    We note that if two vertices share the second coordinate, their color classes are always identical.
\end{definition}
\begin{figure}
\centering
\begin{tikzpicture}[
    >=Stealth,
    vertex/.style={
        circle,
        fill=black,
        inner sep=0pt,
        minimum size=4pt
    },
    label/.style={
        font=\ttfamily\large,
        draw=none
    },
    horizontal/.style={
        -{Stealth[length=2mm]},
        very thick,
        shorten >=3pt,
        shorten <=3pt
    },
    vertical/.style={
        -{Stealth[length=2mm]},
        very thick,
        shorten >=3pt,
        shorten <=3pt
    },
    diagonal/.style={
        -{Stealth[length=2mm]},
        ultra thick,
        shorten >=3pt,
        shorten <=3pt
    }
]

\foreach \i in {0,...,5} {
    \foreach \j in {0,1} {
        \node[vertex,color=red, minimum size=4.5pt] (n-\i-\j) at (1.25*\i, -1-1*\j) {};
    }
    \foreach \j in {2,3} {
        \node[vertex, rectangle, color=green!80!black] (n-\i-\j) at (1.25*\i, -1-1*\j) {};
    }
    \foreach \j in {4,5} {
        \node[vertex,regular polygon, regular polygon sides=3, minimum size=6pt, color=blue] (n-\i-\j) at (1.25*\i, -1-1*\j) {};
    }
}

\foreach \i in {0,...,4} {
    \foreach \j in {0,...,5} {
        \draw[horizontal] (n-\i-\j) -- (n-\the\numexpr\i+1\relax-\the\numexpr\j\relax);
    }
    \foreach \j in {0,...,4} {
        \draw[diagonal] (n-\i-\j) -- (n-\the\numexpr\i+1\relax-\the\numexpr\j+1\relax);
    }
    \foreach \j in {1,...,5} {
        \draw[diagonal] (n-\i-\j) -- (n-\the\numexpr\i+1\relax-\the\numexpr\j-1\relax);
    }
}

\end{tikzpicture}
\caption{An example of $\layergrid_{5,5}$ with three color classes.}
\label{fig:colourclasses}
\end{figure}

In the base case of the outer algorithm, we consider only a small subgraph of the grid graph.
However, as paths between two vertices of the subgraph do not have to be contained in the subgraph itself (see Figure~\ref{fig:path_outside}), we must consider a slightly larger subgraph, in particular the single color class becomes larger by a multiplicative factor.
In order to ensure that we really only add the paths between the original smaller color classes, we set the weight of all edges from the first and last layer of the graph that are not incident to vertices of the smaller color class.
This new oracle is then used for the inner algorithm and also for the final algorithm to ensure we count only relevant paths.

\begin{figure}
    \centering
    \begin{tikzpicture}[
    >=Stealth,
    vertex/.style={
        circle,
        fill=black,
        inner sep=0pt,
        minimum size=4pt
    },
    label/.style={
        font=\ttfamily\large,
        draw=none
    },
    horizontal/.style={
        -{Stealth[length=2mm]},
        shorten >=3pt,
        shorten <=3pt
    },
    vertical/.style={
        -{Stealth[length=2mm]},
        shorten >=3pt,
        shorten <=3pt
    },
    diagonal/.style={
        -{Stealth[length=2mm]},
        thick,
        shorten >=3pt,
        shorten <=3pt
    },
    highlighted/.style={
        -{Stealth[length=8pt, width=8pt]},
        line width=3.5pt,
        color=red!80!black,
        shorten >=3pt,
        shorten <=3pt
    },
    highlightedvertex/.style={
        circle,
        fill=red!80!black,
        inner sep=0pt,
        minimum size=6pt
    },
]

\foreach \i in {0,...,5} {
    \foreach \j in {0,...,4} {
        \node[vertex] (n-\i-\j) at (1*\i, -1-1*\j) {};
    }
}

\foreach \i in {0,...,4} {
    \foreach \j in {0,...,4} {
        \draw[horizontal] (n-\i-\j) -- (n-\the\numexpr\i+1\relax-\the\numexpr\j\relax);
    }
    \foreach \j in {0,...,3} {
        \draw[diagonal] (n-\i-\j) -- (n-\the\numexpr\i+1\relax-\the\numexpr\j+1\relax);
    }
    \foreach \j in {1,...,4} {
        \draw[diagonal] (n-\i-\j) -- (n-\the\numexpr\i+1\relax-\the\numexpr\j-1\relax);
    }
}

\draw[dashed, rounded corners=5pt]
    (0.6,-2.6) rectangle (4.4,-4.4);
    
    \draw[highlighted] (n-1-2) -- (n-2-1);
    \draw[highlighted] (n-2-1) -- (n-3-2);
    \draw[highlighted] (n-3-2) -- (n-4-2);
    \node[highlightedvertex] at (1, -3) {};
    \node[highlightedvertex] at (4, -3) {};

\end{tikzpicture}
    \caption{The larger red vertices are a part of the induced rectangle, but the bold red path connecting them leaves the induced rectangle.}
    \label{fig:path_outside}
\end{figure}

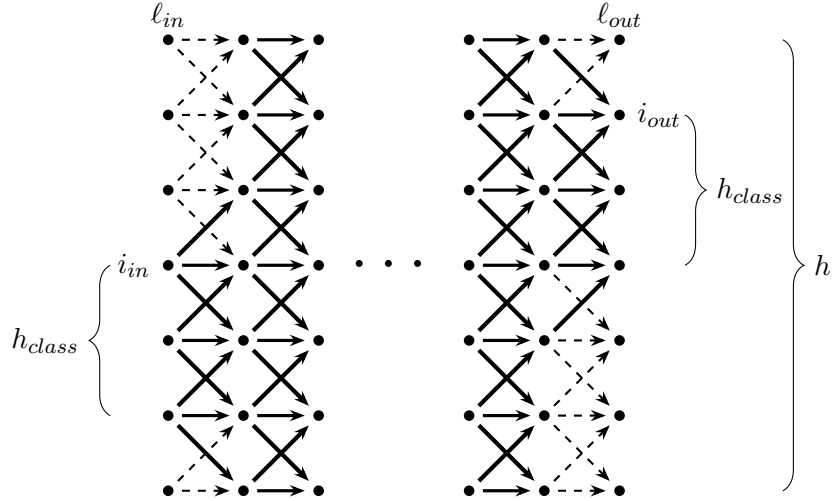
\begin{figure}
    \centering
    \begin{tikzpicture}[
    >=Stealth,
    vertex/.style={
        circle,
        fill=black,
        inner sep=0pt,
        minimum size=4pt
    },
    label/.style={
        font=\ttfamily\large,
        draw=none
    },
    horizontal/.style={
        -{Stealth[length=2mm]},
        very thick,
        shorten >=3pt,
        shorten <=3pt
    },
    vertical/.style={
        -{Stealth[length=2mm]},
        very thick,
        shorten >=3pt,
        shorten <=3pt
    },
    diagonal/.style={
        -{Stealth[length=2mm]},
        ultra thick,
        shorten >=3pt,
        shorten <=3pt
    }
]

\foreach \i in {0,1,2,4,5,6} {
    \foreach \j in {0,...,6} {
        \node[vertex] (n-\i-\j) at (1*\i, -1-1*\j) {};
    }
}

\foreach \i in {1,4} {
    \foreach \j in {0,...,6} {
        \draw[horizontal] (n-\i-\j) -- (n-\the\numexpr\i+1\relax-\the\numexpr\j\relax);
    }
    \foreach \j in {0,...,5} {
        \draw[diagonal] (n-\i-\j) -- (n-\the\numexpr\i+1\relax-\the\numexpr\j+1\relax);
    }
    \foreach \j in {1,...,6} {
        \draw[diagonal] (n-\i-\j) -- (n-\the\numexpr\i+1\relax-\the\numexpr\j-1\relax);
    }
}

\foreach \i in {0} {
    \foreach \j in {0,...,2,6} {
        \draw[horizontal, thick, dashed] (n-\i-\j) -- (n-\the\numexpr\i+1\relax-\the\numexpr\j\relax);
    }
    \foreach \j in {0,...,2} {
        \draw[diagonal, thick, dashed] (n-\i-\j) -- (n-\the\numexpr\i+1\relax-\the\numexpr\j+1\relax);
    }
    \foreach \j in {1,...,2,6} {
        \draw[diagonal, thick, dashed] (n-\i-\j) -- (n-\the\numexpr\i+1\relax-\the\numexpr\j-1\relax);
    }
    \foreach \j in {3,4,5} {
        \draw[horizontal] (n-\i-\j) -- (n-\the\numexpr\i+1\relax-\the\numexpr\j\relax);
    }
    \foreach \j in {3,4,5} {
        \draw[diagonal] (n-\i-\j) -- (n-\the\numexpr\i+1\relax-\the\numexpr\j+1\relax);
    }
    \foreach \j in {3,4,5} {
        \draw[diagonal] (n-\i-\j) -- (n-\the\numexpr\i+1\relax-\the\numexpr\j-1\relax);
    }
}

\foreach \i in {5} {
    \foreach \j in {0,4,5,6} {
        \draw[horizontal, thick, dashed] (n-\i-\j) -- (n-\the\numexpr\i+1\relax-\the\numexpr\j\relax);
    }
    \foreach \j in {3,4,5} {
        \draw[diagonal, thick, dashed] (n-\i-\j) -- (n-\the\numexpr\i+1\relax-\the\numexpr\j+1\relax);
    }
    \foreach \j in {1,5,6} {
        \draw[diagonal, thick, dashed] (n-\i-\j) -- (n-\the\numexpr\i+1\relax-\the\numexpr\j-1\relax);
    }
    \foreach \j in {1,2,3} {
        \draw[horizontal] (n-\i-\j) -- (n-\the\numexpr\i+1\relax-\the\numexpr\j\relax);
    }
    \foreach \j in {0,1,2} {
        \draw[diagonal] (n-\i-\j) -- (n-\the\numexpr\i+1\relax-\the\numexpr\j+1\relax);
    }
    \foreach \j in {2,3,4} {
        \draw[diagonal] (n-\i-\j) -- (n-\the\numexpr\i+1\relax-\the\numexpr\j-1\relax);
    }
}

\node at (3,-4) {\Huge $\cdots$};

\draw[
  decorate,
  decoration={brace, amplitude=8pt, raise=2pt}
]
  (-0.7,-6) -- (-0.7,-4)
  node[midway, left=10pt] {$\hclass$};

\draw[
  decorate,
  decoration={brace, amplitude=8pt, raise=2pt}
]
  (6.8,-2) -- (6.8,-4)
  node[midway, right=10pt] {$\hclass$};

\draw[
  decorate,
  decoration={brace, amplitude=8pt, raise=2pt}
]
  (8.1,-1) -- (8.1,-7)
  node[midway, right=10pt] {$h$};

\node[above=1pt] at (n-0-0) {$\lin$};
\node[above=1pt] at (n-6-0) {$\lout$};
\node[left=3pt] at (n-0-3) {$\iin$};
\node[right=3pt] at (n-6-1) {$\iout$};

\end{tikzpicture}
    \caption{The behavior of the masking oracle in~\Cref{lem:oracle_transform_subgraph}. The dashed edges have their edge weights set to $0$ by the mask.}
    \label{fig:masking_oracle}
\end{figure}

\begin{lemma}[Weight Oracle for $\outeralg$]\label[lemma]{lem:oracle_transform_subgraph}
    Let $\weightoracle$ be a weight oracle for the graph $\layergrid_{n,n}$ with edge weights $\mu: E(\layergrid_{n,n})\rightarrow \N$.
    Consider an induced subgraph $H=H(\lin, \lout, \itop, h, \iin, \iout, \hclass)$ of $\layergrid_{n,n}$ such that $V(H)=\{(i,j): \lin\leq i \leq \lout, \itop\leq j< \itop+h\}$ with edge weights $\mu':V(H)\times V(H)\rightarrow \N$ such that
    \[\mu'((i,j), (i+1,k))=\begin{cases}
        0 & \text{if } (i =\lin \land (j < \iin \lor \iin + \hclass \leq j))\,\lor\\
          & \phantom{\text{if }} (i+1 =\lout \land (k < \iout \lor \iout + \hclass \leq k)),\\
        \mu((i,j), (i+1,k)) & \text{otherwise.}
    \end{cases}\]
    Then there exists a weight oracle $\weightoracle'$ for the graph $H(\lin, \lout, \itop, h, \iin, \iout, \hclass)$ and edge weights $\mu'$ such that $\weightoracle'$ requires only additive $\bigO(\log(n))$ space and $\bigO(\log(n))$ time on top of the requirements for $\weightoracle$.
\end{lemma}

\begin{proof}
    The algorithm exactly follows the formula defining $\mu'$ -- we first check the conditions and then we either return zero immediately, or we defer to the original oracle.
    In either case, we only do constant number of comparisons and additions, and we need a constant number of variables for those, yielding our logarithmic space and time additive overhead.
\end{proof}

\begin{lemma}\label[lemma]{lem:layered_outer}
    There exists a catalytic subroutine $\outeralg^{\ell}(\bigRin, \bigRout; \lin, \lout, \lmid; \cin, \cout)$, where $\bigRin, \bigRout \in \F^{n/C}$ are any sections of the catalytic space and $\lin, \lout, \lmid$ are three $\log(n)$-bit free space variables, such that for $\lout-\lin \leq 2^\ell\cdot n/C$ and given a weighted graph $\layergrid_{n,n}$ with $C$ color classes, where the weights are given by an oracle $\weightoracle$ that on input $e$ outputs the weight of the edge $e$ modulo $p$, if the initial contents of $\bigRin$ and $\bigRout$ were $\tauin$ and $\tauout$ respectively, then $\bigRin$ remains unchanged, and $\bigRout$ becomes $\tauout + \tauin\cdot \calW$ with respect to the color class $\cin$ in layer $\lin$ and color class $\cout$ in layer $\lout$.

    This procedure uses free space $\ell\cdot(\log C + 4) + \bigO(\log |\F| + \log n)$, catalytic space $\ell \cdot \frac{n}{C}\cdot \log |\F| + \bigO(n/C \cdot \log(n/C)\cdot \log(|\F|))$, time $(4C)^{\ell}\cdot \bigO(n^3) + 5\cdot \sum_{i=0}^{\ell-1}(4C)^{i}$ and $(4C)^{\ell}\cdot \bigO((\frac{n}{C})^3)$ oracle queries.
\end{lemma}
\begin{proof}
    As in~\cref{lem:layered_inner}, we assume that $n, h, \lout-\lin$ are all powers of two. (This can be removed at the cost of a few bits of free space per level of recursion.)

    We prove this using induction.
    For $\outeralg^{0}(\bigRin, \bigRout; \lin, \lout, \lmid; \cin, \cout)$, we want to apply \Cref{lem:layered_inner} using $\inneralg^{\log(n/C)}(\bigRinprime, \bigRoutprime; \lin, \lout, \lmid; \itop, h)$.
    
    First, we note that if $|\cin-\cout|\geq 2$, we do not need to do anything, as having two color classes whose indices differ by at least two means that any path between the two vertices must change height by one per each step, and there are only $n/C$ possible steps, while there are at least $n/C+2$ steps needed, and thus no path between these layers can exist. Therefore, in such case, the subroutine returns immediately with no changes.
    
    If $|\cin-\cout|\leq 1$, then we need to consider the subgrid containing the all paths from the left color class $\cin$ to the right color class $\cout$.
    However, we might need to consider more than the two subgrids containing these two color classes to cover all possible paths, as some path could jump above or below one of the color classes and then return back into the subgrid as shown in Figure~\ref{fig:path_outside}.
    To cover the subgrid above, we set $\itop:=\max(\min(\cin,\cout)-1,0)\cdot n/C$ and $h:=4n/C$ to preserve our assumption on powers of two.
    Using the notation from \Cref{lem:oracle_transform_subgraph}, we now focus on the subgraph $H(\lin, \lout, \itop, h, \cin\cdot n/C, \cout\cdot n/C, n/C)$.
    The description of the subgraph comprises the seven $\log(n)$-bit integers.
    We pass these as parameters together with the oracle $\weightoracle$ to the oracle from~\Cref{lem:oracle_transform_subgraph}, and so we obtain an oracle $\weightoracle'$ with only additive logarithmic overhead.
    We then allocate the two catalytic register vectors of size $4n/C$ and \emph{swap} the currently used $\bigRin, \bigRout$ with their corresponding parts in $\bigRinprime, \bigRoutprime$.
    Then, we run $\inneralg^{\log(n/C)}(\bigRinprime, \bigRoutprime; \lin, \lout, \lmid; \itop, h)$ with the oracle $\weightoracle'$.
    After finishing the computation, we revert the swap of $\bigRin, \bigRout$ with their corresponding parts in $\bigRinprime, \bigRoutprime$.

    The values added to $\bigRout$ are correct by the correctness of~\Cref{lem:layered_inner}.
    The rest of the $\bigRinprime, \bigRoutprime$ remains the same after the computation, as $\bigRinprime$ does not change, and the only region with changes in $\bigRoutprime$ is the region that is swapped with $\bigRout$, and thus after the final swap, $\bigRoutprime$ remains unchanged after all operations.

    Analyzing time, we need $4^{\log(n/C)} \cdot (\bigO(\frac{n}{C}\cdot \polylog{|\F|}) + 20) =\bigO((\frac{n}{C})^3\cdot \polylog{|\F|})$ steps.
    The required free space is $(\log(n/C)-1)\cdot 4 + \bigO(\log |\F| + \log n)=\bigO(\log |\F| + \log n)$, and we use $(\ell-1)\cdot h\cdot \log(|\F|)=\bigO(n/C \cdot \log(n/C)\cdot \log(|\F|))$ catalytic space.
    Finally, we need $4^{\ell-1}\cdot \bigO(h) = \bigO((\frac{n}{C})^3)$ weight oracle queries.

    We are now ready to define $\outeralg^{\ell}(\bigRin, \bigRout; \lin, \lout, \lmid; \cin, \cout)$ for $\ell\geq 1$.
    We use the first currently unused catalytic vector $\bigRmid \in \F^{n/C}$, and then we perform the following for each of the color classes $c_{mid}\in [C]$:
    \begin{enumerate}
        \item $\lmid = \frac{\lin+\lout}{2}$
        \item $\outeralg^{\ell-1}(\bigRmid, \bigRout; \lmid, \lout, \lin; c_{mid}, \cout)^{-1}$ \\ \makebox[1in]{} \codecomment{$\bigRout[\vout]\minuseq \sum\limits_{\vmid\in \bigvmid}\bigRmid[\vmid]\cdot \calW_{\vmid,\vout}\forall \vout$}
        \item $\lin = 2\cdot \lmid-\lout$
        \item $\outeralg^{\ell-1}(\bigRin, \bigRmid; \lin, \lmid, \lout; \cin, c_{mid})$\codecomment{$\bigRmid[\vmid]\pluseq \sum\limits_{\vin\in \bigvin}\bigRin[\vin]\cdot \calW_{\vin,\vmid}\forall \vmid$}
        \item $\lout = 2\cdot \lmid-\lin$
        \item $\outeralg^{\ell-1}(\bigRmid, \bigRout; \lmid, \lout, \lin; c_{mid}, \cout)$ \\ \makebox[1in]{ } \codecomment{$\bigRout[\vout]\pluseq \sum\limits_{\vmid\in \bigvmid}\bigRmid[\vmid]\cdot \calW_{\vmid,\vout}\forall \vout$}
        \item $\lin = 2\cdot \lmid-\lout$
        \item $\outeralg^{\ell-1}(\bigRin, \bigRmid; \lin, \lmid, \lout; \cin, c_{mid})^{-1}$\codecomment{$\bigRmid[\vmid]\minuseq \sum\limits_{\vin\in \bigvin}\bigRin[\vin]\cdot \calW_{\vin,\vmid}\forall \vmid$}
        \item $\lout = 2\cdot \lmid-\lin$
    \end{enumerate}

    The total time spent is by induction $4C \cdot ((4C)^{\ell-1}\cdot \bigO(n^3) + 5\cdot \sum_{i=1}^{\ell-2}(4C)^{i}) + 5= (4C)^{\ell}\cdot \bigO(n^3) + 5\cdot \sum_{i=0}^{\ell-1}(4C)^{i}$, while we use $4C \cdot ((4C)^{\ell-1}\cdot \bigO((\frac{n}{C})^3))=(4C)^{\ell}\cdot \bigO((\frac{n}{C})^3)$ weight oracle queries.
    
    Regarding free space used, we only need a variable for iterating over all $C$ color classes which requires $\log C$ space, and 4 bits for instruction count.
    Thus, using induction, we require space $\log C + 4 + (\ell-1)\cdot(\log C + 4) + \bigO(\log |\F| + \log n) = \ell\cdot(\log C + 4) + \bigO(\log |\F| + \log n)$.
    For catalytic space, we allocate a new field vector that requires $\frac{n}{C}\cdot \log |\F|$ bits, and by induction, we require $\frac{n}{C}\cdot \log |\F| + (\ell-1) \cdot \frac{n}{C}\cdot \log |\F| + \bigO(n/C \cdot \log(n/C)\cdot \log(|\F|)) = \ell \cdot \frac{n}{C}\cdot \log |\F| + \bigO(n/C \cdot \log(n/C)\cdot \log(|\F|))$.

    The correctness of the computation follows identical ideas from~\Cref{lem:main_subroutine}, while proving that we maintain the correct values of $\lin, \lout, \lmid$ is identical to the proof in~\Cref{lem:layered_inner}.
\end{proof}

We now have all the tools necessary to state and prove the main result of this section: we can compute the sum of all weights of $s,t$-paths modulo a prime $p$ for two particular vertices $s,t$.
Our approach will use the path-flow propagation of~\Cref{lem:layered_outer} with the ideas of~\cite{cook_pyne_26_efficient} so that we are able to extract the desired value using four runs of the algorithm.

\begin{theorem}\label{thm-gridalg}
    There exists a catalytic algorithm that, given a prime $p=\poly(n)$, a graph $\layergrid_{n,n}$ with weights $\mu: E(\layergrid_{n,n})\rightarrow\N$ given as an oracle $\weightoracle$ that returns the weights modulo $p$, and two vertices $u, v\in V(\layergrid_{n,n})$ described by their indices in the layer such that the distance between $u$ and $v$ is a power of two, computes the value $\calW_{u,v}=\sum_{P\, u,v\text{-path}}\mu(P)\bmod{p}$.

    The algorithm runs in time $\bigO(n^{3+\varepsilon})$ using free space $\bigO(\log n)$, catalytic space $\frac{n}{2^{\Omega(\sqrt{\log n})}}$ and $\bigO(\frac{n^{3+\varepsilon}}{2^{3\sqrt{\log n}}})$ weight oracle queries.
\end{theorem}
\begin{proof}
    Our goal is to use the $\outeralg$ procedure from~\Cref{lem:layered_outer}.
    
    To do so, we need the following: a suitable choice of the parameter $C$, two catalytic vector registers of length $n/C$, three $\log(n)$-bit free space variables and two more $\log(C)$-bit variables for the description of the first and last color class.

    We set $C:= 2^{\lceil\sqrt{(\varepsilon/2.1)\log n}\rceil}$ and allocate the two register vectors $\bigRin, \bigRout$ according to Remark~\ref{rmk:catalytic_mod_p}.
    Let $\ell_u, i_u$ and $\ell_v, i_v$ be the indices of $u$ and $v$ respectively.
    We can easily compute $\ell_v-\ell_u$, from which we can compute a suitable $\ell$ to use the right level of recursion for $\outeralg^\ell$.
    Similarly, we allocate the three $\log(n)$-bit integers for layer counting $\lin, \lout, \lmid$, where $\lin:=\ell_u$ and $\lout:=\ell_v$.
    The value of $\lmid$ may be arbitrary as it is immediately replaced in the procedure.
    We also allocate the two values $\cin, \cout$ for the color classes of $u$ and $v$ respectively.

    Finally, we also use~\Cref{lem:oracle_transform_subgraph} on the subgraph $H(\lin, \lout, 0, n, i_u, i_v, 1)$ to obtain an oracle $\weightoracle'$ which sets the weight zero to every edge coming from the layer $\ell_u$ unless it is incident with $u$, and to every edge coming to the layer $\ell_v$ unless it is incident with $v$, and it mirrors the original oracle $\weightoracle$ otherwise.

    As observed in the proof of~\cref{lem:layered_outer}, we can also run the $\outeralg^{\ell}$ procedure in reverse with flipped signs to subtract the values.
    Therefore, we perform the following procedure:
    \begin{enumerate}
        \item Run $\outeralg^{\ell}(\bigRin, \bigRout; \lin, \lout, \lmid; \cin, \cout)$\codecomment{$\bigRout[\vout]\pluseq \sum\limits_{\vin\in \bigvin}\bigRin[\vin]\cdot \calW_{\vin,\vout}\forall \vout$}
        \item Record the value $\bigRout[v]$ to free space as $r_1$.
        \item Run $\outeralg^{\ell}(\bigRin, \bigRout; \lin, \lout, \lmid; \cin, \cout)^{-1}$\codecomment{$\bigRout[\vout]\minuseq \sum\limits_{\vin\in \bigvin}\bigRin[\vin]\cdot \calW_{\vin,\vout}\forall \vout$}
        \item Update $\bigRin[u]\pluseq 1$.
        \item Run $\outeralg^{\ell}(\bigRin, \bigRout; \lin, \lout, \lmid; \cin, \cout)$\codecomment{$\bigRout[\vout]\pluseq \sum\limits_{\vin\in \bigvin}\bigRin[\vin]\cdot \calW_{\vin,\vout}\forall \vout$}
        \item Record the value $\bigRout[v]$ to free space as $r_2$.
        \item Run $\outeralg^{\ell}(\bigRin, \bigRout; \lin, \lout, \lmid; \cin, \cout)^{-1}$\codecomment{$\bigRout[\vout]\minuseq \sum\limits_{\vin\in \bigvin}\bigRin[\vin]\cdot \calW_{\vin,\vout}\forall \vout$}
        \item Update $\bigRin[u]\minuseq 1$.
        \item Return $r_2-r_1\pmod{p}$.
    \end{enumerate}

    To argue correctness, we need to show that the catalytic space has been restored to the original state and that the value returned is indeed correct.
    First, we note that line 3 undoes all changes performed by line 1 and line 7 undoes all changes made by line 5. Then, the only two changes are lines 4 and 8, which add and subtract 1 from $\bigRin[u]$ respectively, and these two lines again negate each other.
    Therefore, the catalytic space has the same content at the beginning and the end of the procedure.

    Next, we consider correctness of the result.
    Let $\tauin$ be the initial value in $\bigRin[u]$ and let $\tauout$ be the initial value in $\bigRout[v]$.
    By~\Cref{lem:layered_outer} and our setup for the weights, $r_1 = \tauout + \tauin\cdot \calW_{u,v}$ and $r_2 = \tauout + (\tauin+1)\cdot \calW_{u,v}$.
    Thus, $r_2-r_1 = \calW_{u,v}$, and performing all the arithmetic modulo $p$ ensures that we indeed return $\calW_{u,v}\bmod p$.

    We now analyze time and space complexity.
    We split the analysis into two steps: first, we consider the resources needed for $\outeralg^\ell$.
    In the worst case, the depth necessary for $\outeralg^\ell$ to work is $\ell = \lceil\sqrt{\log n}\rceil$.
    Then, \Cref{lem:layered_outer} gives us that $\outeralg^\ell$ uses free space $\ell\cdot(\log C + 4) + \bigO(\log |\F| + \log n)=\lceil\sqrt{(\varepsilon/2.1)\log n}\rceil\cdot (\log(2^{\lceil\sqrt{(\varepsilon/2.1)\log n}\rceil})+4) + \bigO(\log p + \log n) = \bigO(\log n)$.
    The catalytic space used is $\ell \cdot \frac{n}{C}\cdot \log |\F| + \bigO(n/C \cdot \log(n/C)\cdot \log(|\F|)) = \lceil\sqrt{(\varepsilon/2.1)\log n}\rceil\cdot \frac{n}{2^{\lceil\sqrt{(\varepsilon/2.1)\log n}\rceil\cdot}}\cdot \log n + \bigO(\frac{n}{2^{\lceil\sqrt{(\varepsilon/2.1)\log n}\rceil\cdot}} \cdot \log(\frac{n}{2^{\lceil\sqrt{(\varepsilon/2.1)\log n}\rceil\cdot}})\cdot \log n) = \frac{n}{2^{\Omega(\sqrt{\log n})}}$.
    The total time spent is $(4C)^{\ell}\cdot \bigO(n^3) + 5\cdot \sum_{i=0}^{\ell-1}(4C)^{i}\leq(4\cdot 2^{\lceil\sqrt{(\varepsilon/2.1)\log n}\rceil})^{\lceil\sqrt{(\varepsilon/2.1)\log n}\rceil}\cdot \bigO(n^3) + 5\cdot (4\cdot 2^{\lceil\sqrt{(\varepsilon/2.1)\log n}\rceil})^{\lceil\sqrt{(\varepsilon/2.1)\log n}\rceil}=\bigO(n^{3+\varepsilon})$ plus we use $(4C)^{\ell}\cdot \bigO((\frac{n}{C})^3)=\bigO(\frac{n^{3+\varepsilon}}{2^{3\sqrt{\log n}}})$ oracle queries.

    Regarding the rest of the program, we need to allocate a constant number of $\bigO(\log(n))$-bit free space variables, and two vectors of $\bigO(\log(p))$-bit field elements of length $n/C=\frac{n}{2^{\sqrt{(\varepsilon/2.1)\log n}}}$ in the catalytic space.
    
    Thus, the total free space usage is $\bigO(\log n)$.
    We also use $\frac{n}{2^{\Omega(\sqrt{\log n})}} + \frac{n}{2^{\sqrt{(\varepsilon/2.1)\log n}}}\cdot \log(p)=\frac{n}{2^{\Omega(\sqrt{\log n})}}$ catalytic space.
    The running time is $4\cdot \bigO(n^{3+\varepsilon}) + \bigO(\polylog{n})=\bigO(n^{3+\varepsilon})$, as each of the lines 2, 4, 6, 8, 9 requires logarithmic time in the worst case and the other lines are covered by the time complexity of $\outeralg^\ell$.
    We only use oracle queries in $\outeralg^\ell$, so the number of weight oracle queries is $\bigO(\frac{n^{3+\varepsilon}}{2^{3\sqrt{\log n}}})$.
\end{proof}

We remark that the assumption on the distance between $u$ and $v$ being a power of two is only done for simplicity of exposition.
It is possible to support arbitrary distances at the cost of $\bigO(1)$ bits of free space per level of recursion, which would amount only to $\bigO(\log n)$ additional free space required.

\subsection{Computing edit distance, LCS and discrete Fr\'{e}chet distance}

Our goal is now to establish Theorem \ref{thm-edit} and then Theorem \ref{thm-frechet}. 
To compute edit distance and longest common subsequence we will need to calculate a certain $N$-bit number $x$ where $N=\bigO(n^2)$.
We do not have enough space to store $x$ or the intermediate results so we cannot compute it directly.
Instead we will compute $x$ modulo various $\bigO(\log n)$-bit primes, i.e. in Chinese remainder representation,  
and then access individual bits of $x$
using the algorithm of \cite{HesseAB02} from Proposition \ref{prop-CRR}.

Next we will explain how to compute edit distance of two strings using Theorem \ref{thm-gridalg}.
We will reduce computation of edit distance to computing a total path weight value in a particular grid graph.
From the binary representation of the total path weight value we can determine the edit distance of the two strings.
The grid graph will be a simple modification of the usual edit distance graph.
We can compute the total path weight value in the grid graph modulo any $\bigO(\log n)$-bit prime using Theorem \ref{thm-gridalg}.
Using Proposition \ref{prop-CRR} we can recover the binary representation of the total path weight value and hence compute the edit distance of the two strings.

\begin{figure}
    \centering
    \begin{tikzpicture}[
    >=Stealth,
    vertex/.style={
        circle,
        fill=black,
        inner sep=0pt,
        minimum size=4pt
    },
    contdot/.style={
        circle,
        fill=black,
        inner sep=0pt,
        minimum size=2pt
    },
    label/.style={
        font=\ttfamily\large,
        draw=none
    },
    horizontal/.style={
        -{Stealth[length=2mm]},
        very thick,
        shorten >=3pt,
        shorten <=3pt
    },
    vertical/.style={
        -{Stealth[length=2mm]},
        very thick,
        shorten >=3pt,
        shorten <=3pt
    },
    diagonal/.style={
        -{Stealth[length=2mm]},
        ultra thick,
        shorten >=3pt,
        shorten <=3pt
    }
]

\foreach \i in {0,...,2} {
    \foreach \j in {0,...,2} {
        \node[vertex] (n-\i-\j) at (1*\i, -4+1*\j) {};
    }
}
\foreach \i in {3,4} {
    \foreach \j in {3,4} {
        \node[vertex] (n-\i-\j) at (1*\i, -4+1*\j) {};
    }
}

\node at (2, -4.5) {\large$\ed_{n,n}$};

\foreach \i in {0,4,6,8} {
    \foreach \j in {2} {
        \node[vertex] (n2-\i-\j) at (6+0.75*\i, -0.5-0.75*\j) {};
    }
}

\foreach \i in {1,3,7} {
    \foreach \j in {1,...,3} {
        \node[vertex] (n2-\i-\j) at (6+0.75*\i, -0.5-0.75*\j) {};
    }
}

\foreach \i in {2} {
    \foreach \j in {0,...,4} {
        \node[vertex] (n2-\i-\j) at (6+0.75*\i, -0.5-0.75*\j) {};
    }
}

\node (n2-tl) at (6, 1) {};
\node (n2-br) at (12, -5) {};

\node at (9, -5.5) {\large$\layergrid_{2n,2n}$};

\node at (5, -2) {\LARGE $\leadsto$};



\begin{scope}[on background layer]
    \foreach \i in {0,...,1} {
        \foreach \j in {0,...,2} {
            \draw[horizontal] (n-\i-\j) -- (n-\the\numexpr\i+1\relax-\j);
        }
    }
    
    \foreach \i in {0,...,2} {
        \foreach \j in {0,...,1} {
            \draw[vertical] (n-\i-\j) -- (n-\i-\the\numexpr\j+1\relax);
        }
    }
    
    \foreach \i in {0,...,1} {
        \foreach \j in {0,...,1} {
            \draw[diagonal] (n-\i-\j) -- (n-\the\numexpr\i+1\relax-\the\numexpr\j+1\relax);
        }
    }
\end{scope}

\draw[horizontal] (n-3-3) -- (n-4-3);
\draw[horizontal] (n-3-4) -- (n-4-4);
\draw[vertical] (n-3-3) -- (n-3-4);
\draw[vertical] (n-4-3) -- (n-4-4);
\draw[diagonal] (n-3-3) -- (n-4-4);

\begin{scope}
    \clip (0,-2) rectangle (3,0);
    \draw[dashed]
        (n-0-0) rectangle (n-4-4);
\end{scope}
\begin{scope}
    \clip (2,-4) rectangle (4,-1);
    \draw[dashed]
        (n-0-0) rectangle (n-4-4);
\end{scope}

\node[contdot] at (0.5, -1.25) {};
\node[contdot] at (0.5, -1.5) {};
\node[contdot] at (0.5, -1.75) {};
\node[contdot] at (1.5, -1.25) {};
\node[contdot] at (1.5, -1.5) {};
\node[contdot] at (1.5, -1.75) {};
\node[contdot] at (2.25, -2.5) {};
\node[contdot] at (2.5, -2.5) {};
\node[contdot] at (2.75, -2.5) {};
\node[contdot] at (2.25, -1.75) {};
\node[contdot] at (2.5, -1.5) {};
\node[contdot] at (2.75, -1.25) {};
\node[contdot] at (2.25, -3.5) {};
\node[contdot] at (2.5, -3.5) {};
\node[contdot] at (2.75, -3.5) {};

\node[vertex, fill=green, minimum size=8pt] at (n-0-0) {};
\node[vertex, fill=blue!80!red, minimum size=8pt] at (n-4-4) {};

\foreach \i in {0,...,3} {
    \foreach \j in {2} {
        \draw[horizontal] (n2-\i-\j) -- (n2-\the\numexpr\i+1\relax-\the\numexpr\j\relax);
    }
}
\foreach \i in {6,7} {
    \foreach \j in {2} {
        \draw[horizontal] (n2-\i-\j) -- (n2-\the\numexpr\i+1\relax-\the\numexpr\j\relax);
    }
}
\foreach \i in {1,2} {
    \foreach \j in {1,3} {
        \draw[horizontal] (n2-\i-\j) -- (n2-\the\numexpr\i+1\relax-\the\numexpr\j\relax);
    }
}
\draw[diagonal] (n2-0-2) -- (n2-1-3);
\draw[diagonal] (n2-0-2) -- (n2-1-1);
\draw[diagonal] (n2-1-1) -- (n2-2-0);
\draw[diagonal] (n2-1-1) -- (n2-2-2);
\draw[diagonal] (n2-1-3) -- (n2-2-2);
\draw[diagonal] (n2-1-3) -- (n2-2-4);
\draw[diagonal] (n2-2-4) -- (n2-3-3);
\draw[diagonal] (n2-2-2) -- (n2-3-3);
\draw[diagonal] (n2-2-0) -- (n2-3-1);
\draw[diagonal] (n2-2-2) -- (n2-3-1);
\draw[diagonal] (n2-3-1) -- (n2-4-2);
\draw[diagonal] (n2-3-3) -- (n2-4-2);
\draw[diagonal] (n2-6-2) -- (n2-7-1);
\draw[diagonal] (n2-6-2) -- (n2-7-3);
\draw[diagonal] (n2-7-1) -- (n2-8-2);
\draw[diagonal] (n2-7-3) -- (n2-8-2);    
\begin{scope}
    \clip (7.5,1) rectangle (11.25,-5);
    \draw[dashed, rotate=-45]
        (n2-0-2) rectangle (n2-8-2);    
\end{scope}
\draw (n2-tl) rectangle (n2-br);

\node[contdot] at (9.5, -2) {};
\node[contdot] at (9.75, -2) {};
\node[contdot] at (10, -2) {};
\node[contdot] at (9, -1.25) {};
\node[contdot] at (9.25, -1) {};
\node[contdot] at (9.5, -0.75) {};
\node[contdot] at (8.25, -0.5) {};
\node[contdot] at (8.5, -0.25) {};
\node[contdot] at (8.75, -0) {};
\node[contdot] at (9, -2.75) {};
\node[contdot] at (9.25, -3) {};
\node[contdot] at (9.5, -3.25) {};
\node[contdot] at (8.25, -3.5) {};
\node[contdot] at (8.5, -3.75) {};
\node[contdot] at (8.75, -4) {};

\node[vertex, fill=green, minimum size=8pt] at (n2-0-2) {};
\node[vertex, fill=blue!80!red, minimum size=8pt] at (n2-8-2) {};

\end{tikzpicture}
    \caption{The transformation from $\ed_{n,n}$ into $\layergrid_{2n,2n}$. The source and sink vertices of $\ed_{n,n}$ and $\layergrid_{2n,2n}$ which correspond to each other are shown in the same color. There are no edges in $\layergrid_{2n,2n}$ outside the dashed square.}
    \label{fig:grid_transform}
\end{figure}
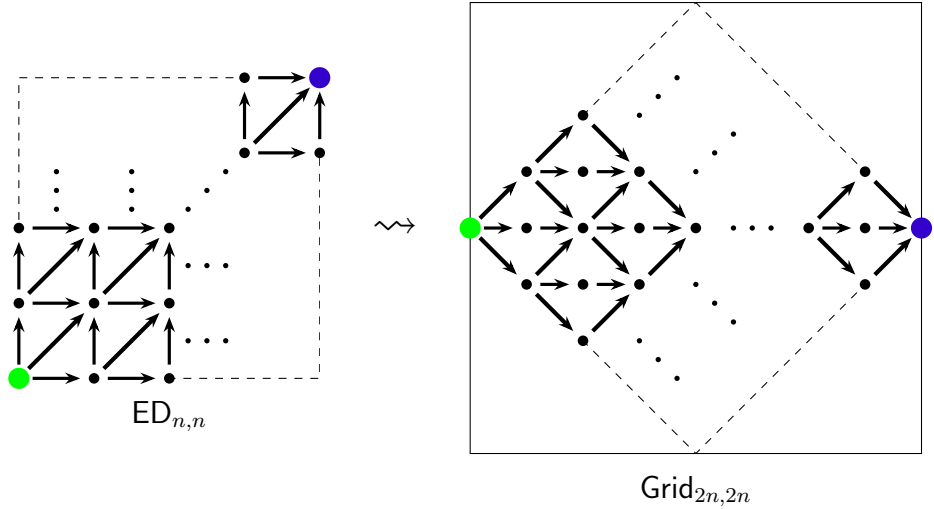

For two strings $x$ and $y$ of length $n$, define the edit distance graph $\ed_{x,y}$ to have vertex set $\{0,\dots,n\} \times \{0,\dots,n\}$,
horizontal edges $(i,j) \rightarrow (i+1,j)$ of cost 1, for $0\le i < n$ and $0\le j \le n$,
vertical edges $(i,j) \rightarrow (i,j+1)$ of cost 1, for $0\le i \le n$ and $0\le j < n$,
and diagonal edges $(i,j) \rightarrow (i+1,j+1)$, for $0\le i,j < n$,
where the cost of the edge is 0 if $x_{i+1}=y_{j+1}$ and it is 1 otherwise.

It is well known that the edit distance of $x$ and $y$ is the minimum cost of a path from $(0,0)$ to $(n,n)$ in $\ed_{x,y}$ (see e.g. \cite{AndoniKO10,ChakrabortyGK16}).
(Here by the cost of a path we understand the sum of the costs of the edges on the path as opposed to the weight of a path which is the product of the weights.)
We will embed $\ed_{x,y}$ into $\layergrid_{2n, 2n}$ so that the minimum cost of a path from $(0,0)$ to $(n,n)$ in $\ed_{x,y}$
can be determined from a particular total path weight in  $\layergrid_{2n, 2n}$.
To embed $\ed_{x,y}$ into $\layergrid_{2n, 2n}$ we will rotate the graph by 45 degrees clockwise (Fig. \ref{fig:grid_transform})
and split each diagonal edge into two by adding an auxiliary vertex.

We will define weights of the edges in $\layergrid_{2n, 2n}$ as follows.
Each vertex $(i,j)\in \{0,\dots,2n\}^2$ of $\layergrid_{2n, 2n}$, where $i = j+1 \pmod{2}$, corresponds to an auxiliary diagonal vertex,
and its only outgoing edge $(i,j) \rightarrow (i+1,j)$ has weight 1; the other two out-going edges have weight 0.
For each vertex $(i,j)\in \{0,\dots,2n\}^2$ where $i = j \pmod{2}$ we let $i'=(i-j +n)/2$ and $j'=(i+j-n)/2$.
If $(i',j')\in \{0,\dots,n\} \times \{0,\dots, n\}$ then the vertex $(i,j)$ in $\layergrid_{2n, 2n}$ corresponds to a vertex $(i',j')$ in $\ed_{x,y}$.
We set weights of edges leaving $(i,j)$ in $\layergrid_{2n, 2n}$ as follows:
\begin{enumerate}
    \item If $0 \le i'<n$ then the edge $(i,j)\rightarrow (i+1,j-1)$ has weight $16^n$.
    \item If $0 \le j'<n$ then the edge $(i,j)\rightarrow (i+1,j+1)$ has weight $16^n$.
    \item Finally, if $0 \le i',j'<n$ then the edge $(i,j)\rightarrow (i+1,j)$ has weight $1$ if $x_{i'+1} = y_{j'+1}$, and weight $16^n$ otherwise.
\end{enumerate}
The weight of all other edges in $\layergrid_{2n, 2n}$ is set to zero.
We call $\layergrid_{2n,2n}$ with this particular weighting scheme $\layergrid_{x, y}$.

Hence, vertex $(0,0)$ of $\ed_{x,y}$ corresponds to the vertex $(0,n)$ in $\layergrid_{x,y}$,
and  vertex $(n,n)$  corresponds to the vertex $(2n,n)$ in $\layergrid_{x,y}$.
There is a one-to-one correspondence between paths from $(0,0)$ to $(n,n)$ in $\ed_{x,y}$ and paths from $(0,n)$ to $(2n,n)$ in  $\layergrid_{x,y}$ that go only through non-zero-weight edges.
A path of cost $w$ from $(0,0)$ to $(n,n)$ in $\ed_{x,y}$ corresponds to a path of weight $16^{nw}$ in  $\layergrid_{x,y}$.
(A cost-one edge in $\ed_{x,y}$ corresponds to an edge with weight $16^n$ in  $\layergrid_{x,y}$ and a cost-zero edge in $\ed_{x,y}$ corresponds to an edge of weight $1$ in  $\layergrid_{x,y}$. 
The {\em auxiliary} diagonal edges of $\ed_{x,y}$ have weight 1 in $\layergrid_{x,y}$ so they don't affect the product of weights.) 

Given $n$ and $\bigO(\log n)$-bit prime $p$ in binary, we can compute $16^n \bmod p$ in $\bigO(\log n)$-space by repeated squaring and multiplications. 
Hence, it is straightforward to implement a weight oracle $\weightoracle$ for $\layergrid_{x,y}$ as required by Theorem \ref{thm-gridalg}.
The oracle will use space $\bigO(\log n)$ for its internal computation and it will need an access to the two input strings $x$ and $y$.
If the length of the two strings is not a power of two the oracle can extend the graph $\layergrid_{x,y}$ by extra vertices 
to make it of desired size $\layergrid_{2n',2n'}$ (where $n< n'< 2n$ is a power of two) 
where there will be a path of weight 1 leading from the original target vertex $(2n,n)$ to the new target vertex $(2n',n)$, 
and all other new edges will be assigned weight 0.

Notice, there are at most $3^{2n}=9^n$ paths from $(0,0)$ to $(n,n)$ in $\ed_{x,y}$ as each vertex has out-degree at most 3 and each path is of length at most $2n$.
Hence, there are at most $9^n < 16^n$ paths of non-zero weight from $(0,n)$ to $(2n,n)$ in  $\layergrid_{x,y}$.
For $w\in \{0,\dots, n\}$, let $n_w$ denote the number of paths of cost $w$ from $(0,0)$ to $(n,n)$ in $\ed_{x,y}$.
The total path weight from $s=(0,n)$ to $t=(2n,n)$ in $\layergrid_{x,y}$
corresponds to $\calW_{s,t}=\sum_{w=0}^{2n} 16^{nw} \cdot n_w$.
Since each $n_w < 16^{n}$, if we represent $\calW_{s,t}$ in base $16^n$ each digit will correspond to one of $n_w$'s.
Furthermore, if $\calW_{s,t}$ is represented in binary then bits $4n(w+1)-1,\dots, 4nw$ represent $n_w$ in binary
(where the lowest order bit has index 0).
Thus, if we determine $\calW_{s,t}$ we can determine the first non-zero $n_w$ for $w=0,\dots,2n$ by inspecting the binary representation of $\calW_{s,t}$.
(This requires $\bigO(\log n)$-bits of additional space.)
The binary representation of $\calW_{s,t}$ has $N=(2n+1)4n=8n^2+4n$ bits.
We can access the individual bits of $\calW_{s,t}$ using Proposition \ref{prop-CRR} that will make at most $\bigO(N^2)=\bigO(n^4)$ queries to $\calW_{s,t} \bmod p_i$, $i<N$,
where each of the queries can be answered by running the algorithm from Theorem \ref{thm-gridalg} on our $\layergrid_{x,y}$.
Hence, we can compute the edit distance of $x$ and $y$.

Overall the algorithm requires only $\bigO(\log n)$-space in addition to the space needed by the algorithm of Theorem \ref{thm-gridalg} to answer $\calW_{s,t} \bmod p_i$ queries.
This establishes Theorem~\ref{thm-edit} for edit distance.

We note that we do not establish the precise time complexity, as it also depends on the time complexity of~\Cref{prop-CRR} which we do not attempt to optimize here.

The same approach also works for the weighted versions of edit distance if the weights are $\bigO(\log n)$-bit integers with appropriate changes to the weighting scheme.
The number of bits in the number would then be $\bigO(n^2\cdot w_{\text{max}})$, where $w_{\text{max}}=\poly(n)$ is the largest possible weight, which would still yield a polynomial time algorithm.

\medskip
Longest common subsequence of two strings can be computed in a similar manner.
Instead of the graph $\ed_{x,y}$ one considers an almost identical graph $\lcs_{x,y}$ where horizontal and vertical edges have cost zero 
and each diagonal edge $(i,j)\rightarrow(i+1,j+1)$ has cost one if $x_{i+1} = y_{j+1}$ and cost zero otherwise.
The length of the longest common subsequence of $x$ and $y$ corresponds to the maximum cost of a path from $(0,0)$ to $(n,n)$.
Similarly to the edit distance we embed $\lcs_{x,y}$ into $\layergrid_{2n,2n}$, and compute the value of the total path weight $\calW_{s,t}$, for $s=(0,n)$ to $t=(2n,n)$.
From the binary representation of $\calW_{s,t}$ we can determine the maximum cost of a path from $(0,0)$ to $(n,n)$ that appears in the original graph $\lcs_{x,y}$.
This implies Theorem~\ref{thm-edit} for longest common subsequence.

\medskip
Computing the discrete Fr\'{e}chet distance of two point sequences $x$ and $y$ can be done using a similar technique, while using a well-known reduction that is already implicit in the algorithm by Eiter and Mannila~\cite{EiterM94}.
By our assumption on the distance oracle, we know that all the pairwise distances of points from the two sequences are from $\{0,\dots,\poly(n)\}$.
For $m \in \{0,\dots, \poly(n)\}$ we can define a graph $\frechet_{x,y,m}$ on a vertex set $\{0,\dots,n\} \times \{0,\dots,n\}$
with horizontal, vertical and diagonal edges like in the edit distance graph 
except that we remove edges leading into each vertex $(i,j)$ where the distance between $x_i$ and $y_j$ is more than $m$.
Such vertices will become unreachable.
The cost of all the remaining edges will be one.
In this graph $\frechet_{x,y,m}$, there is a path from $(0,0)$ to $(n,n)$ if and only if the discrete Fr\'{e}chet distance of $x$ and $y$ is at most $m$.
Again we can embed the graph into $\layergrid_{2n,2n}$ while assigning weight 0 to missing edges and weight one to edges that are present. 
We can use the algorithm from Theorem \ref{thm-gridalg} to determine the total path weight from $(0,n)$ into $(2n,n)$ modulo a prime.
As all paths will have weight one the total path weight counts the number of paths from $(0,n)$ into $(2n,n)$.
To determine whether the discrete Fr\'{e}chet distance of $x$ and $y$ is at most $m$ we just have to check whether the total path weight is non-zero.
The total path weight is zero if and only if its value is zero modulo the first $4n$ primes.
Thus, we can check whether there is a path from $(0,0)$ to $(n,n)$ in $\frechet_{x,y,m}$ in time $\bigO(n^{4+\varepsilon})$.

Using $\bigO(\log n)$ iterations of the algorithm, we perform binary search on $m$ from $\{0,\ldots,\poly(n)\}$ to find the smallest $m$ with the property that the discrete Fr\'{e}chet distance of $x$ and $y$ is at most $m$. 
This is the discrete Fr\'{e}chet distance of $x$ and $y$.
We conclude Theorem~\ref{thm-frechet}.

\bibliography{refs}
\bibliographystyle{alpha}

\appendix

\section{Catalytic Technicalities}
\label{sec:catalytic_technicalities}

In this section, we deal with the remaining technical details of our catalytic algorithms. \\

\noindent
{\bf Catalytic vectors $\bmod$ $p$.}\label{rmk:catalytic_mod_p}
Our algorithm views pieces of its catalytic memory as vectors of registers $r_1,\dots,r_m$ storing values from $\{0,\dots,p-1\}$.
If we allocate $w=\lceil \log p \rceil$ bits to each of the registers in a naive way 
the initial value in some of the registers could be out of the allowed range when $p$ is not a power of two.
Previous works \cite{BCKLS14,CGMPS25} addressed this issue using techniques which lead to large time overhead when accessing individual registers. 
Here we present a different technique which builds on the previous ones which is suitable for machines that have random access to the catalytic tape.
Our technique incurs $\bigO(\log n \log p)$ overhead to access a register given by its index.

Our solution will use $8mw$ bits of the catalytic tape to store the registers.
It has three phases: pre-processing phase taking time $\bigO(m \cdot \polylog{m+p})$, working phase when we can work with the registers, 
and post-processing phase taking time $\bigO(m \cdot \polylog{m+p})$ which returns the catalytic tape to its original state 
(assuming all the registers are already returned to the same value as they had at the beginning of the working phase.)

Assume that $m$ is a power of two. 
Choose some constant $C\ge 1$, set $M=m^C$.
The registers are allocated as follows. 
We will divide the allocated space into groups of $b=8 \log M$ cells, each cell having $w$ bits.
The $i$-th group of cells, $i=1,\dots,m/\log M$, will accommodate $\log M$ registers $r_{1+(i-1)\log M},\dots,r_{i \log M}$.
Within the group, the $j$-th register will use the $j$-th cell that contains a  value in range $\{0,\dots,p-1\}$. 
Hence, depending on the exact details of the computational model locating the cell storing a given register takes time at most $\bigO(\log M \log p)$ to inspect cells in a given group.

The only issue is that initially a group might contain fewer than $\log M$ cells with a value  in the range $\{0,\dots,p-1\}$.
(Cells with a value in that range keep the value in that range as register cannot be assigned a value outside of its range.)
We call such groups {\em unhealthy}.

Unhealthy groups will be modified during the pre-processing stage so that they will be usable in the working stage.
During the pre-processing stage we will create a linked list of the unhealthy cells 
so that during the post-processing phase we will be able to return them to their original contents.
The pointers of the linked list will use space obtained by compressing each unhealthy group so only $\bigO(\log m)$ extra bits will be needed to store the pointer to the first unhealthy group.

During the pre-processing phase if a group is unhealthy, then all but at most $b/8$ cells in the group have the most significant bit set to one.
We will flip the most significant bit of the first $b/4$ cells in the group.
Since $b/4-b/8 \ge \log M$ the group is now suitable to accommodate $\log M$ registers within its first $b/4$ cells.
The latter $3b/4$ cells will be compressed by removing their most significant bit which will free $3b/4$ bits.
The $3b/4$ bits will be used to store the pointer to the next unhealthy group ($\log M$ bits)
and information about the deleted most significant bits ($5b/8$ bits).

The most significant bits are encoded as triples: for each triple we have one indicator bit saying whether it consists of only 1's
and if not we specify the whole triple. 
The indicators occupy $b/4$ bits and the explicit triples take at most $3b/8$ bits.
Together this is at most $5b/8$ bits so there are at least $3b/4-5b/8 = b/8 = C \log m \ge \log m$ bits available for the pointer.

In the post-processing phase we go through the list of unhealthy groups and undo the compression and the bit flips. 

Notice, the unhealthy groups provide extra unused space when $C>1$.
This extra space can be used if we need to manage multiple vectors in the catalytic memory at once without incurring overhead $\bigO(\log n)$ bits
for a separate pointer to unhealthy cell list in each of the vectors.  

The above construction uses $8mw$ bits to store about $mw$ bits of information. 
One can easily improve the construction to require space only $(2+\epsilon)nw$
by a better compression of unhealthy groups and by making the groups constant time larger.
This would affect the time of pre- and post-processing by poly-logarithmic factors and the access time by a constant factor.

\end{document}